\def\draft{0}
\def\sigconf{0}
\def\big{0}
\def\anon{0}
\def\extabst{0}
\def\coverpageandmarkers{0}
\ifnum\big=1
    \documentclass[final,17pt]{extarticle} \usepackage{fullpage}
\else
    \ifnum\sigconf=1
        \documentclass[sigconf,final]{acmart}
    \else
        \ifnum\draft=1
           \documentclass[11pt,draft,a4paper]{article}
        \else
            \documentclass[a4paper,onecolumn,11pt,accepted=2020-06-25]{quantumarticle}
        \fi
    \fi
\fi

\newif\ifab 
\abfalse

\newif\iffv 
\fvtrue

\pdfoutput=1

\ifnum\sigconf=0\usepackage{authblk}\fi \usepackage{silence} 
\usepackage[modulo]{lineno}
\usepackage{amsthm,amssymb,amsfonts,latexsym,color,paralist,url,ifdraft,mathrsfs,thm-restate,comment,bbold,booktabs}

\usepackage[utf8]{inputenc} \usepackage[autostyle=false, style=english]{csquotes} \MakeOuterQuote{"}
\usepackage[T1]{fontenc}
\ifnum\sigconf=0
    \usepackage[draft=false,pageanchor]{hyperref}
    
    \usepackage[final]{graphicx}
\fi
\ifdraft{\usepackage{showlabels}}{} 

\RequirePackage{etex}
\usepackage [ lambda,advantage , operators , sets , adversary , landau , probability , notions , logic ,ff, mm,primitives , events , complexity , asymptotics , keys]{cryptocode}
\renewcommand{\verify}{\ensuremath{\mathsf{verify}}}

\theoremstyle{plain}
  
\ifnum\sigconf=0
    \newtheorem{theorem}{Theorem}

    \newtheorem{definition}[theorem]{Definition}

\fi 
  
  \newtheorem{fact}[theorem]{Fact}
  \newtheorem*{theorem*}{Theorem}
  \newtheorem*{lemma*}{Lemma}
  \newtheorem*{corollary*}{Corollary}
  \newtheorem*{proposition*}{Proposition}
  \newtheorem*{claim*}{Claim}
  
  \theoremstyle{definition}
  
  \theoremstyle{remark}

  \theoremstyle{plain}

\ifdraft{\newcommand{\authnote}[3]{{\color{#3} {\bf  #1:} #2}}}{\newcommand{\authnote}[3]{}}

\newcommand{\ket}[1]{\vert #1 \rangle}

\newcommand{\tensor}{\otimes}

\ifdraft{\linenumbers}{}

\usepackage{algorithm,algorithmicx,algpseudocode}
\usepackage[normalem]{ulem}

\usepackage{tikz}
\usetikzlibrary{arrows,chains,matrix,positioning,scopes}
\makeatletter
\tikzset{join/.code=\tikzset{after node path={%
\ifx\tikzchainprevious\pgfutil@empty\else(\tikzchainprevious)%
edge[every join]#1(\tikzchaincurrent)\fi}}}
\makeatother
\tikzset{>=stealth',every on chain/.append style={join},
         every join/.style={->}}
\tikzstyle{labeled}=[execute at begin node=$\scriptstyle,
   execute at end node=$]
\ifnum\sigconf=0
    \usepackage[capitalise]{cleveref} 
\fi
\usepackage{autonum}

\newcommand{\ab}[1]{}  

\newcommand{\onote}[1]{\authnote{Or}{#1}{blue}}
\newcommand{\anote}[1]{\authnote{Andrea}{#1}{red}}

\newcommand{\keygen}{\ensuremath{\mathsf{key\textit{-}gen}}}
\newcommand{\sign}{\ensuremath{\mathsf{sign}}}
\newcommand{\ver}{\ensuremath{\textsf{verify-bolt}}}
\newcommand{\mint}{\ensuremath{\mathsf{mint}}}
\newcommand{\verifysig}{\ensuremath{\mathsf{verify\textit{-}sig}}}

\newcommand{\gencertificate}{\ensuremath{\mathsf{gen\textit{-}certificate}}}
\newcommand{\verifycertificate}{\ensuremath{\mathsf{verify\textit{-}certificate}}}

\newcommand{\sigforge}{\ensuremath{\mathsf{Sig\textit{-}Forge}}}

\newcommand{\sabotagemoney}{\ensuremath{\mathsf{Sabotage\textit{-}Money}}}
\newcommand{\sabotagecertificate}{\ensuremath{\mathsf{Sabotage\textit{-}Certificate}}}
\newcommand{\sabotagesignature}{\ensuremath{\mathsf{Sabotage\textit{-}Signature}}}

\newcommand{\gen}{\ensuremath{\mathsf{Gen}}}

\ifab
\usepackage[margin=1.1in]{geometry}
\else
\usepackage[margin=1in,includefoot,footskip=30pt]{geometry}
\fi

\usepackage[T1]{fontenc}
\usepackage[utf8]{inputenc}
\usepackage{amsmath,amsfonts,amsthm,amssymb}
\usepackage{mathtools}
\usepackage{subeqnarray}
\usepackage{setspace}
\usepackage{graphics,graphicx,color}
\usepackage{url}
\usepackage{enumerate}
\usepackage{float} 
\usepackage{multirow}
\usepackage{hhline}
\usepackage{braket}
\usepackage{dsfont} 
\usepackage{authblk}
\usepackage{multirow}
\usepackage{hhline}
\usepackage[makeroom]{cancel}

\usepackage[margin=1.2in]{geometry}
\newtheorem{prop}{Proposition}

\usepackage[scr=boondoxo,scrscaled=1.05]{mathalfa}

\newcommand{\beq}{\begin{eqnarray}}
\newcommand{\eeq}{\end{eqnarray}}

\newcounter{mparcnt}
\renewcommand\themparcnt{\arabic{mparcnt}}
\ifnum\coverpageandmarkers=1
    \newcommand\mpar[1]{\refstepcounter{mparcnt}\marginpar{\textbf{Marker \themparcnt} #1}}
\else
    \newcommand\mpar[1]{}
\fi
\crefname{mparcnt}{Marker}{Markers}
\begin{document}
\title{A Quantum Money Solution to the Blockchain Scalability Problem
\ifnum\extabst=1 
\footnote{Extended abstract. This work supersedes two previous independent works, \cite{Col19} and \cite{Or2019}.}
\fi
}
\ifnum\anon=0
    \ifnum\sigconf=0
        \author[1]{Andrea Coladangelo}
        \author[2]{Or Sattath}
        \affil[1]{Computing and Mathematical Sciences, Caltech}
        \affil[2]{Computer Science Department, Ben-Gurion University}
    \else
        \author{Andrea Coladangelo}
        \affiliation{%
        \istitution{Computing and Mathematical Sciences,	
        Caltech}
        \country{USA}}
        \author{Or Sattath}
        \affiliation{%
        \institution{Computer Science Department, Ben-Gurion University}
        \country{Israel}}
    \fi
\else
    \ifnum\sigconf=0
        \author{}
    \fi
\fi
\ifnum\coverpageandmarkers=1
We would like to begin by thanking the reviewers for their in-depth review, and the editor’s comments. These include important points, which we mostly agree with, and we hope that we addressed them. We believe this significantly improved the quality of the manuscript. On top of these, the manuscript includes significant additions, which are partly the result of a merge with a second paper by Or Sattath. 

More specifically, we formalize a new property of a quantum lightning scheme which we call “bolt-to-signature” capability. This is reminiscent of one-time digital signatures. We provide a provably secure construction of a lightning scheme with such a property (assuming a secure lightning scheme which might not satisfy this). The construction is based on the hash-and-sign paradigm as well as Lamport’s signature scheme. We discuss extensively how this resolves some existing practical issues with our payment system in an elegant way, but we also believe that this could be of independent interest. The revised version of the manuscript also includes an additional section, which provides an extensive comparison of our payment system with Bitcoin, and other second-layer solutions, like Bitcoin’s Lightning Network. 

Before proceeding, we believe that although most of Reviewer 2’s comments are to the point, we fear that Reviewer 2 has misunderstood one crucial aspect about our construction. We have clarified this misunderstanding in our answers below, and edited the manuscript to make this crucial aspect clear. We have also addressed the other major concern of Reviewer 2, in our replies below, and in the manuscript where required. We suspect that these additional explanations and edits will lead Reviewer 2 to a reassessment of the submission.

We will start by summarizing the reviewers’ feedback:
Both reviewers agree that the aim of the paper is well motivated, and has scientific value. In simple terms, the aim is to “inherit the trustlessness from the blockchain protocol and the scalability and transaction speed from quantum money” (quote from Reviewer 1). 

Reviewer 1 had several technical comments, which we will touch upon in more detail later. 

Reviewer 2 raised several major concerns. Before addressing these, we comment on one of Reviewer 2’s assessments about the paper.
\vspace{3mm}

\emph{The main contribution seems to be a clever idea of how one can combine a (very strong, too strong?) ideal ledger with Zhandry’s quantum money. In essence, a payee has to find a pre-image x of a hash s = H(x) to claim their money. This preimage is encoded in a quantum state $\ket{\psi}$ which the payer sends to the payee. A measurement then allows the payee to find x. Due to no-cloning (and the properties of Zhandry’s scheme), the payer cannot generate multiple copies of $\ket{\psi}$, and therefore cannot double spend.}
\vspace{3mm}

The first major concern raised by the referee is tightly related to this point as well, so we address these two points two together. We fear that Reviewer 2 has missed a crucial point about our manuscript (probably due to an unclear explanation, which we have extensively clarified.). 

The reviewer noticed that one way to use our payment system is so that Alice, the payer, sends to Bob, the payee, the quantum state. At that point, Bob could “redeem” the quantum money state back to a classical coin. We fear that the referee got the impression that this is the only intended way the system could be used. Our point is that Bob could spend the quantum money to Charlie, the latter could then spend it to Daniel, etc: importantly none of these transactions require posting any new transaction on the blockchain. 
The reviewer wondered whether we assume that some operations are faster than others. We agree that this point was not as clear as it should have been. In short, the point is that some operations are indeed slow -- namely, “write” operations which alter the state of the ledger should be considered as slow (in Bitcoin, this would be the order of 10 minutes, and in Ethereum ~ 10 seconds). Yet, reading information that is already in the ledger is fast -- in systems such as Bitcoin and Ethereum today, this is governed by how long it takes to retrieve an entry from a database of hundreds of gigabytes in size -- which could typically take milliseconds. Therefore, the setup phase, in which a coin is transformed to a quantum money is a slow process, as it requires altering the state of the ledger. But once this is done, payments are fast since verifying a payment only requires reading data from the ledger, and locally applying the quantum lightning verification procedure to the received quantum banknote. This is similar in spirit to the opening of a Bitcoin Lightning Network channel, (which requires adding data to the blockchain), and then transacting using the Lightning Network is fast (since nothing more than read access to the Bitcoin network is needed). We have now added an extensive analysis in Section 8 of the different tradeoffs between a cryptocurrency such as Bitcoin, a “second-layer” solution such as the Lightning Network, and our quantum payment system. We have also clarified which are the “fast” and which are the “slow” operations on the ledger - see Marker \ref{mar:fast_vs_slow}.

We address the remaining two major concerns of Reviewer 2.
\begin{enumerate}
    \item The reviewer states that we discuss the UC framework, and the importance of it in order to argue about composability. In something that seems like a sharp contradiction, the security proofs are not UC. We agree that this discrepancy was not explained, and that it might even seem misleading.
    We have now made extra clear that we do not claim that our scheme is universally composable. Nonetheless, we use the UC framework for two reasons. First and foremost, if someone comes up with a real world implementation of the ideal ledger functionality we discuss, and proves that it UC-securely realizes our ideal ledger functionality, we could then replace the ideal functionality in our protocol with its real world implementation, and the security guarantee that we proved would hold verbatim (by the composition theorem of the UC framework). We feel that this alone is a good enough reason to operate in the UC framework (despite the extra notation). Secondly, we feel that a formal definition of the ideal functionality, as well as the way adversaries and parties are modelled is necessary in order to make any formal (Non-UC) security statement. The UC framework provides a convenient formalism for this. We have added a discussion about these points on \cref{mar:GUC_simplifies,mar:GUC_could_be_used,mar:not_composable}.
    \item The reviewer points out that it is crucial that the verification does not perturb the quantum money. The reviewer wondered why this property is not in the definition. Indeed, this property should have been part of the definition, and was added -- see \cref{mar:no_perturbation}.
    
    We also discuss a related issue, which was not discussed in the original submission, denoted security against sabotage,  see \cref{mar:motivation_sabotage} for the motivation, and \cref{sec: sabotage} for more details. 
\end{enumerate}

We address several other minor issues raised by Reviewer 1:
\begin{enumerate}
    \item \emph{In the ledger, the adversary directly communicates with the ideal functionality. I’m a bit unsure whether directly notifying “the adversary” about additional parties is the most natural choice here. I see how this is a convenient formalization of things, but it is very counterintuitive that the ideal functionality actively helps any adversary. To me it seems that it would be closer to a real-world ledger, and thus more natural, as well as more in the spirit of the UC framework, to fix the number of parties in the beginning. Subsequently, some or all of them can register a PID. That way, the adversary can corrupt an unregistered party and register them. Subsequently, they can retrieve information about the set of players, e.g. by using RetrieveTransaction.}

Yes, we went for the convenient formalization of things, but it could have been as suggested too.
   \item \emph{In the last item of the itemized list in the bullet "execution phase" of the description of smart contracts, I think $\geq$ should be replaced by $>$ and $<$ by $\leq$. That way, the smart contract also terminates when e=``all coins''. Should the behavior of not terminating upon e=``all coins'' be intended, I think this would deserves some discussion.}
   
   Yes, the replacements suggested are correct. We have fixed this (see Figure~\ref{fig: ledger functionality}). Thanks for spotting this. 
\end{enumerate}

\vspace{3mm}

We address the minor comments by referee 2:
\begin{enumerate}
    \item \emph{How are the quantum states sent from payer to payee? On an authenticated (and therefore encrypted) quantum channel as well?}
    
    Indeed, we assume the quantum channel is authenticated. We also assume it is ideal. This is clarified in Marker~\ref{mar:ideal_authenticated}. 
    
    \item \emph{Definition 3 and the following paragraphs refer a lot to $Game_{12}$. There’s a $Game_{1}$ and a $Game_{2}$ just before Def 3, but I can’t find any definition of $Game_{12}$. I’m assuming this is a typo, and is supposed to be $Game_{2}$, because that seems to be the context in which it is used.}
    
    That has been fixed now. We noticed that the notation is cumbersome, and chose to improve it: There is no longer a $Game_1$ -- it is now part of the completeness definition. Furthermore,  $Game_2$ has been renamed ``\textsf{Forge-certificate}'' -- see Marker~\ref{mar:no_perturbation}.

    \item \emph{In the proof of Proposition 1, A is an adversary that can with $Game_2$ ($Game_{12}$?) with non-negligible probability. According to the definition of $Game_2$, A receives (Gen,Ver,H) as input. But in the proof of Proposition 1, $Game_2$ receives (Gen,Ver,H,$\ket{\psi}$) as input.}
    
    All of these issues have been fixed. See the proof of Proposition~\ref{prop: extra property 1}.
    
    \item \emph{In the proof of proposition 1 it says: $\psi’$ must pass Ver with non-negligible probability. What does it mean to ``pass Ver''. Ver seems to just be a map from states to strings, there is no passing condition.}
    
    In the previous version, we considered verification to be “passed” if Ver output a serial number s which was not the symbol ``$\perp$''. We have modified this now to be more natural: Ver takes as input a state and a serial number, and outputs either “accept” (pass) or “reject”.

    \item \emph{The proof of proposition 1 basically concludes by saying that one of the key steps in the proof can be found somewhere in Zha17. This step should be include in the current paper (at least in the appendix), otherwise the proof of Prop 1 is pointless.}
    
    We have clarified and quoted from \cite{zha19} the exact claim needed for our proof -- see Marker \ref{mar:two_preimages}. A similar issue was present for Proposition 2, and we added the required claim -- see Marker \ref{mar:prop2_quote}.
    
    \item \emph{In $Game_1$ the adversary has to produce a state $\ket{\psi}$, but this state isn’t part of the winning condition of $Game_1$, so what is it for?}
    
    See the similar comment after point 2 above.
\end{enumerate}
\fi  
\ifnum\sigconf=0
    \ifdraft{}{\date{}}
    \maketitle
\fi
\sloppy
%
%
\ifnum\extabst=1
\section{Introduction}
Cryptocurrencies, have recently risen to popular attention, the most well-known example being Bitcoin~\cite{Nak08}. 
Informally, a blockchain is a public ledger consisting of a sequence of blocks. Each block typically contains information about a set of transactions, and a new block is appended regularly via a consensus mechanism that involves the parties of a network and no a priori trusted authority. A blockchain is endowed with a native currency which is employed in transactions, and whose basic unit is a ``coin''. The simplest type of transaction is a \textit{payment}, which transfers coins from one party to another. 


A central issue that needs to be resolved for blockchains is \textit{scalability}. This refers to the problem of increasing the throughput of transactions (i.e. transactions per second, which is currently less than 10 in Bitcoin) while maintaining the resources needed for a party to participate in the consensus mechanism approximately constant, and while maintaining security against adversaries that can corrupt constant fractions of the parties in the network. 

In this work, we show that quantum information is inherently well-suited to tackle this problem. We show that a classical blockchain can be leveraged using quantum money to provide a simple solution to the scalability problem.

The main quantum ingredient that we employ is a primitive called \textit{quantum lightning}, formally introduced by Zhandry \cite{zha19} and inspired by Lutomirski et al.'s notion of collision resistant quantum money \cite{lutomirski2009breaking}. In a public-key quantum money scheme, a bank is entrusted with generating quantum states (we refer to these as quantum banknotes) with an associated serial number, and a public verification procedure allows anyone in possession of the banknote to check its validity. Importantly, trust is placed in the fact that the central bank will not create multiple quantum banknotes with the same serial number.
A quantum lightning scheme has the additional feature that no generation procedure, not even the honest one (and hence not even a bank!), can produce two valid money states with the same serial number, except with negligible probability. This opens to the possibility of having a completely decentralized quantum money scheme. However, if the system ought to be trust-less, some issues need to be addressed: most importantly, who is allowed to mint money? Who decides which serial numbers are valid? 

\onote{Removed from extended abstract: We suspect that this is not an isolated example, but that smart contracts could find many other useful applications in quantum cryptography.}




\paragraph{Our contributions}  We design a hybrid classical-quantum payment system that uses quantum states as banknotes and a classical blockchain to settle disputes and to keep track of the valid serial numbers. This is, to the best of our knowledge, the first example of the use of a classical blockchain in combination with quantum cryptographic tools\onote{This is tricky. \cite{Jog16}? I suggest adding the words "provably secure". Though, I don't think we should cite him in the extended abstract.}. Our payment system has the following desirable features:
\begin{itemize}
\item[(i)] It is decentralized, requiring no trust in any single entity.
\item[(ii)] The quantum banknotes enjoy many of the properties of cash, that crypto-currencies do not have. For example, the transaction throughput is unbounded, payments are as quick as quantum communication, and payments do not incur transaction fees.
\item[(iii)] The rightful owner of a quantum banknote can recover the original value, even if the quantum banknote is damaged or lost; An adversary who would try to abuse this feature would be penalized financially. 
\end{itemize}

Our contribution is primarily conceptual, but our treatment is formal. 
We model a blockchain using the Generalized Universal Composability framework (GUC) of Canetti et al. \cite{canetti2007global}, and we formulate an ideal functionality for a blockchain that handles smart contracts. The main reason why such an approach is desirable is that it abstracts the features of blockchains and smart contracts into building blocks that can be utilized to design more complex protocols in a modular way. We emphasize that the security notion that we prove for our payment system in the GUC setting (with access to the ideal smart contracts functionality) is ``one-shot'' and not composable (the latter would be a significantly more ambitious project which we leave as future work).

As a further technical contribution, in order to achieve part (iii), we develop a novel implementation of a one-time digital signature scheme, based on properties of quantum lightning and inspired by Lamport's signature scheme, in which only the user holding the quantum lightning state can sign a message. 


In the full version, we provide two constructions of our payment system: the first is fully formal and works in the idealized GUC setting with access to an ideal functionality for smart contracts. The second is a construction on the Bitcoin blockchain. Its treatment is less rigorous, but it addresses some of the limitations of the idealized setting.

\ifnum\draft=1 Other works \cite{garay2015bitcoin, pass2017analysis} have operated within a similar framework, and have focused on analyzing security of the Bitcoin protocol. 


\fi
\ifnum\draft=1 Other approaches that do not fall into the GUC framework include Other approaches include \cite{kosba2016hawk} and \cite{dziembowski2018fairswap}. 
\fi

We use the framework of universal smart-contracts as an abstraction layer which significantly simplifies exposition and the security analysis.
However, using this framework has a few disadvantages. 
(i) To the best of our knowledge, there is no known proof of a secure realization of an ideal functionality for a ledger supporting smart-contracts by any crypto-currency. Proving that a complicated system such as Bitcoin or Ethereum GUC-securely realize a proposed ledger functionality seems almost impossible.
(ii) Our GUC framework is an idealized one. For example, it is not clear how would one implement the global setup in a system such as Bitcoin; The users are assumed to be online all the time, which is not always the case. 

We provide a more detailed construction which could be realized in a system such as Bitcoin.  For example, we discuss in detail how to implement the global setup, and certain optimizations to save space on the block-chain. This also highlights the fact that we do not need a universal smart contract system, such as Ethereum to implement our construction. This construction might also be easier to follow for those that are more familiar with Bitcoin, but not experts in the UC framework. The main drawback of this construction is that the security arguments in this construction are less rigorous than the ones used in the GUC framework.

\paragraph{A sketch of our payment system}

The main ingredient that we employ is quantum lightning. As suggested by Zhandry, this primitive seems well-suited for designing a decentralized payment system, as it is infeasible for anyone to copy banknotes (even a hypothetical bank). However, since the generation procedure is publicly known, there needs to be a mechanism that regulates the generation of new valid banknotes
. 

We show how stateful smart contracts on a classical blockchain can be used to provide such a mechanism. For instance, they allow to easily keep track of a publicly trusted list of valid serial numbers. We elaborate on this. In our payment system, the native coin of a classical blockchain is used as a baseline classical currency. Any party can spend coins on the classical blockchain to add a serial number of their choice to the list of valid serial numbers. More precisely, they can deposit any amount (of their choice) $d$ of coins into an appropriately specified smart contract and set the initial value of a \textit{serial number} state variable to whatever they wish (presumably the serial number of a quantum banknote that they have just generated locally). They have thus effectively added to the blockchain a serial number associated to a quantum banknote that, in virtue of this, we can think of as having ``acquired'' value $d$. Payments are made by transferring quantum banknotes: party $A$, the payer, transfers his quantum banknote with serial number $s$ to party $B$, the payee, and references a smart contract whose serial number state variable is $s$. Party $B$ then locally verifies that the received banknote is valid. The appeal of such a payment protocol is that transactions only involve the payer and the payee: no consensus mechanism is required to validate the payment, and the payment need not be recorded on the blockchain. The latter is only invoked when a banknote is first generated, and in case of a dispute. Thus, throughput of transactions is no longer a concern. Likewise, long waiting times between when a payment is initiated and when it is confirmed are also no longer a concern since payments are verified immediately by the payee. One can think of the quantum banknotes as providing an off-chain layer that allows for virtually unlimited throughput.

The full payment system includes a mechanism that allows any party in possession of a quantum banknote to recover the coins deposited in the corresponding smart contract (this makes quantum banknotes and the currency on the blockchain in some sense interchangeable). It also includes a mechanism that allows any honest party who has lost or damaged a valid quantum banknote to recover them.

The latter two mechanisms are more tricky to design. In order to achieve these, we formalize an additional property of a quantum lightning scheme, which we call ``banknote-to-certificate capability'', and show that Zhandry's proposed constructions satisfy this property. This asserts informally that it is possible to measure a valid quantum banknote and obtain a classical certificate, but it is impossible to simultaneously hold both a valid banknote and a valid certificate. In other words, the classical certificate acts as a ``proof of destruction'': a certificate which guarantees that no quantum lightning state with the corresponding serial number exists. '' We also introduce a novel implementation of a one-time signature scheme based on the ``banknote-to-certificate capability'', and inspired by Lamport's signature scheme. This protects honest parties who broadcast a classical certificate to the network from having their classical certificate stolen before it is registered on the blockchain. We refer to the full version of the paper for more detail.

\paragraph{Limitations of our solution}
We see two main limitations to our payment system:
\begin{enumerate*}
    \item 
    The technologies will not be available in the foreseeable future.
    \item The constructions known for quantum lightning, which our construction is based upon, rely on non-standard assumptions. 
    \onote{Originally, it was: "Our construction is based on the existence of a secure quantum lightning scheme. There is only one known concrete construction for quantum lightning whose security hinges on a non-standard computational assumption about the multi-collision resistance of certain degree-2 hash functions." But what about the construction of Farhi et al.? See \cref{sec:farhi_et_al}. I changed it to the slightly more concise version above.}
\end{enumerate*}

\bibliographystyle{alphaabbrurldoieprint}
\bibliography{quantum_money_solution_blockchain_scalability}

\else
  \abstract{We put forward the idea that classical blockchains and smart contracts are potentially useful primitives not only for classical cryptography, but for quantum cryptography as well. Abstractly, a smart contract is a functionality that allows parties to deposit funds, and release them upon fulfillment of algorithmically checkable conditions, and can thus be employed as a formal tool to enforce monetary incentives.

In this work, we give the first example of the use of smart contracts in a quantum setting.
We describe a simple hybrid classical-quantum payment system whose main ingredients are a classical blockchain capable of handling stateful smart contracts, and quantum lightning, a strengthening of public-key quantum money introduced by Zhandry \cite{zha19}. Our hybrid payment system employs quantum states as banknotes and a classical blockchain to settle disputes and to keep track of the valid serial numbers. It has several desirable properties: it is decentralized, requiring no trust in any single entity; payments are as quick as quantum communication, regardless of the total number of users; when a quantum banknote is damaged or lost, the rightful owner can recover the lost value. 
}

\vspace{2mm}
\noindent\textbf{Note:} This work supersedes two previous independent works, \cite{Col19} and \cite{Or2019}.

\tableofcontents

\newpage

\section{Introduction}

Cryptocurrencies, along with blockchains and smart contracts, have recently risen to popular attention, the most well-known examples being Bitcoin and Ethereum \cite{Nak08,But14}. 
Informally, a blockchain is a public ledger consisting of a sequence of blocks. Each block typically contains information about a set of transactions, and a new block is appended regularly via a consensus mechanism that involves the parties of a network and no a priori trusted authority. A blockchain is endowed with a native currency which is employed in transactions, and whose basic unit is a ``coin''. The simplest type of transaction is a \textit{payment}, which transfers coins from one party to another. However, more general transactions are allowed, which are known as \textit{smart contracts}. These can be thought of as contracts stored on a blockchain, and whose consequences are executed upon fulfilment of algorithmically checkable conditions.


A central issue that needs to be resolved for blockchains to achieve mass-adoption is \textit{scalability}. This refers to the problem of increasing the throughput of transactions (i.e. transactions per second) while maintaining the resources needed for a party to participate in the consensus mechanism approximately constant, and while maintaining security against adversaries that can corrupt constant fractions of the parties in the network. For example, Bitcoin and Ethereum, can currently handle only on the order of $10$ transactions per second (Visa for a comparison handles about $3500$ per second)\onote{Source? The following source \url{https://usa.visa.com/run-your-business/small-business-tools/retail.html} suggests ~1700 transactions per second, and supports up to 24,000 tps. It seems that this source is rather old ~2010.}. 

In this work, we show that quantum information is inherently well-suited to tackle this problem. We show that a classical blockchain can be leveraged using tools from quantum cryptography (in particular, quantum money) to provide a simple solution to the scalability problem.\footnote{We clarify that this solution only solves the scalability problem for \textit{payment} transactions, and not for the more general \textit{smart contracts} transactions.}

The main quantum ingredient that we employ is a primitive called \textit{quantum lightning}, formally introduced by Zhandry \cite{zha19} and inspired by Lutomirski et al.'s notion of collision resistant quantum money \cite{lutomirski2009breaking}. In a public-key quantum money scheme, a bank is entrusted with generating quantum states (we refer to these as quantum banknotes) with an associated serial number, and a public verification procedure allows anyone in possession of the banknote to check its validity. Importantly, trust is placed in the fact that the central bank will not create multiple quantum banknotes with the same serial number.
A quantum lightning scheme has the additional feature that no generation procedure, not even the honest one (and hence not even a bank!), can produce two valid money states with the same serial number, except with negligible probability. This opens to the possibility of having a completely decentralized quantum money scheme. However, if the system ought to be trust-less, some issues need to be addressed: most importantly, who is allowed to mint money? Who decides which serial numbers are valid? Our solution leverages a (classical) blockchain to address these questions. 

\paragraph{Our contributions} We design a hybrid classical-quantum payment system that uses quantum states as banknotes and a classical blockchain to settle disputes and to keep track of the valid serial numbers. This is, to the best of our knowledge, the first example of the use of a classical blockchain in combination with quantum cryptographic tools\onote{This is tricky. \cite{Jog16}? I suggest adding the words "provably secure". Though, I don't think we should cite him in the extended abstract.}. Our payment system has the following desirable features:
\begin{itemize}
\item[(i)] It is decentralized, requiring no trust in any single entity.
\item[(ii)] Payments involve the exchange of quantum banknotes, and enjoy many of the properties of cash, which cryptocurrencies do not have. For example, transactions are not recorded on the blockchain, and they involve only the payer and the payee. Thus the throughput is unbounded. Payments are as quick as quantum communication, and they do not incur transaction fees.
\item[(iii)] The rightful owner of a quantum banknote can recover the original value, even if the quantum banknote is damaged or lost; an adversary who tried to abuse this feature would be penalized financially. 
\end{itemize}

Our contribution is primarily conceptual, but our treatment is formal: we work within the Generalized Universal Composability framework (GUC) of Canetti et al. \cite{canetti2007global}, and we formulate an ideal functionality for a blockchain that supports smart contracts. Note that we do not prove composable security of our payment system. Instead, we prove a "one-shot" version of security, assuming parties have access to such an ideal functionality. Nonetheless, we find it desirable to work within the GUC framework, and we discuss the reason for this choice in more detail below. We also provide an informal construction of our payment system on the Bitcoin blockchain. Its treatment is less rigorous, but it addresses some ways in which the payment system could be optimized.

As a further technical contribution, in order to achieve part (iii), we formalize a novel property of quantum lightning schemes, which we call ``bolt-to-signature'' capability and is reminiscent of one-time digital signatures. We provide a provably secure construction of a lightning scheme with such a property (assuming a secure lightning scheme which might not satisfy this). The construction is based on the hash-and-sign paradigm as well as Lamport signatures scheme. It allows a user holding a valid quantum lightning state with serial number $s$ to sign a message $\alpha$ \textit{with respect to $s$}. The security guarantees are that no one who does not possess a lightning state with serial number $s$ can forge signatures \textit{with respect to $s$}, and once the lightning state is utilized to produce even a single signature, it will no longer pass the lightning verification procedure. We envision that such a primitive could find applications elsewhere and is of independent interest. 

\paragraph{How to model a blockchain?}
\onote{I copy/pasted the "how to model a blockchain" paragraph, and made lots of changes to it - see below.}

The essential properties of blockchains and smart contracts can be abstracted by modeling them as ideal functionalities in the Universal Composability (UC) framework of Canetti \cite{canetti2001universally}. Such a framework provides both a formal model for multiparty computation, and formal notions of security with strong composability properties. The approach of studying blockchains and smart contracts within such a framework was first proposed by Bentov al. in \cite{bentov2014use}, and explored further in \cite{bentov2017instantaneous}. The main reason why such an approach is desirable is that it abstracts the features of blockchains and smart contracts into building blocks that can be utilized to design more complex protocols in a modular way. The limitation of the works \cite{bentov2014use,bentov2017instantaneous} is that they modify the original model of computation of UC in order to incorporate coins, but they do not prove that a composition theorem holds in this variant. A more natural approach was proposed by Kiayias et al. in~\cite{kiayias2016fair}, which uses the Generalized Universal Composability (GUC) framework by Canetti et al.~\cite{canetti2007global}. 

In this work, we define a ledger functionality in the GUC framework which supports universal smart contracts. We use this as an abstraction layer which significantly simplifies exposition and the security analysis\mpar{\label{mar:GUC_simplifies}}. The downside of this approach is that, to the best of our knowledge, there is no known proof of a GUC-secure realization of an ideal functionality for a ledger supporting smart contracts by any cryptocurrency.\onote{What's the difference between GUC and UC? Why UC is not sufficient for our implication?} Proving that a complicated system such as Bitcoin or Ethereum GUC-securely realizes a simple ledger functionality, even without the support for smart contracts, requires already substantial work \cite{badertscher2017bitcoin}. The upside is that, as soon as one provides a GUC secure realization on some cryptocurrency of our ideal ledger functionality, one can replace the latter in our payment system with its real world implementation\mpar{\label{mar:GUC_could_be_used}}, and the security of our payment system would hold verbatim, by the composition properties of the GUC framework. For this reason, we find the approach very desirable.  Other approaches that do not fall into the GUC framework include \cite{kosba2016hawk} and \cite{dziembowski2018fairswap}.  

We emphasize that we do not prove that our payment system can be composed securely,\mpar{\label{mar:not_composable}} i.e. we do not define an ideal functionality for our payment system and prove a secure realization of it. Rather, the security we prove is only "one-shot". 

\onote{If there exists a real-world GUC realization, then our payment system could be realized as well, and be secure. On the other hand, we do not model our payment system using ideal functionalities, and in particular, it cannot be used in other protocols.}
\ifnum\draft=1 Other works \cite{garay2015bitcoin, pass2017analysis} have operated within a similar framework, and have focused on analyzing security of the Bitcoin protocol. 

Finally, we remark that the security notion that we prove for our payment system in the GUC setting (with access to the ideal smart contracts functionality) is ``one-shot'' and not composable (the latter would be a significantly more ambitious project which we leave as future work).

\fi
\ifnum\draft=1 Other approaches that do not fall into the GUC framework include Other approaches include \cite{kosba2016hawk} and \cite{dziembowski2018fairswap}. 
\fi


\paragraph{A sketch of our payment system}

The main ingredient that we employ is quantum lightning. As suggested by Zhandry, this primitive seems well-suited for designing a decentralized payment system, as it is infeasible for anyone to copy banknotes (even a hypothetical bank). However, since the generation procedure is publicly known, there needs to be a mechanism that regulates the generation of new valid banknotes (to prevent parties from continuously generating banknotes). 

We show how stateful smart contracts on a classical blockchain can be used to provide such a mechanism. For instance, they allow to easily keep track of a publicly trusted list of valid serial numbers. We elaborate on this. In our payment system, the native coin of a classical blockchain is used as a baseline classical currency. Any party can spend coins on the classical blockchain to add a serial number of their choice to the list of valid serial numbers. More precisely, they can deposit any amount (of their choice) $d$ of coins into an appropriately specified smart contract and set the initial value of a \textit{serial number} state variable to whatever they wish (presumably the serial number of a quantum banknote that they have just generated locally). They have thus effectively added to the blockchain a serial number associated to a quantum banknote that, in virtue of this, we can think of as having ``acquired'' value $d$. Payments are made by transferring quantum banknotes: party $A$, the payer, transfers his quantum banknote with serial number $s$ to party $B$, the payee, and references a smart contract whose serial number state variable is $s$. Party $B$ then \emph{locally} verifies that the received quantum banknote is valid. This completes the transaction. Notice that this does not involve interaction with any other third party, and only requires read access to the blockchain. Of course, the same quantum banknote can be successively spent an unlimited number of times in the same manner, without ever posting any new transaction on the blockchain. The latter is only invoked when a banknote is first generated, and in case of a dispute. Thus, throughput of transactions is no longer a concern. Likewise, long waiting times between when a payment is initiated and when it is confirmed (which are typically in the order of minutes -- for example, it is recommended to accept a Bitcoin transaction after 6 confirmations, which takes an hour in expectation) are also no longer a concern since payments are verified immediately by the payee. One can think of the quantum banknotes as providing an off-chain layer that allows for virtually unlimited throughput. We give an extended comparison of our payment system with other off-chain solutions, like Bitcoin's Lightning Network, in Section \ref{sec: comparison}. 

The full payment system includes two additional optional features\mpar{}:
\begin{itemize}
    \item[(i)] A mechanism that allows any party in possession of a quantum banknote to recover the coins deposited in the corresponding smart contract. Together with the mechanism outlined earlier, this makes quantum banknotes and coins on the blockchain in some sense interchangeable: one can always convert one into the other by publishing a single transaction on the blockchain.
    \item[(ii)] A mechanism that allows any honest party who has lost or damaged a valid quantum banknote to change the \textit{serial number} state variable of the corresponding smart contract to a fresh value of their choice. 
\end{itemize}

These two additional desirable features are realizable as long as the quantum lightning scheme employed satisfies an additional property. We formalize this property, which we call "bolt-to-certificate capability", and show that Zhandry's proposed constructions satisfy this property.
The latter asserts informally that it is possible to measure a valid quantum banknote and obtain a classical certificate, but it is impossible to simultaneously hold both a valid banknote and a valid certificate. In other words, the classical certificate acts as a ``proof of destruction'': a certificate which guarantees that no quantum lightning state with the corresponding serial number exists. We also introduce a novel implementation of a one-time signature scheme based on the ``bolt-to-certificate capability'', and inspired by Lamport's signature scheme. This protects honest parties who broadcast a classical certificate to the network from having their classical certificate stolen before it is registered on the blockchain. 

We emphasize that the bolt-to-certificate capability of the lightning scheme is only required in order to achieve the two additional features (i) and (ii) above, but is not necessary in order to achieve the basic functionality of our payment system described above.

\paragraph{Limitations of our solution}
We see two main limitations to our payment system:
\begin{itemize}
    \item 
    The technologies needed to implement our payment system will not be available in the foreseeable future. For instance, our payment system would require the ability of each party to store large entangled quantum states for extended periods of time. It would also require that each party be able to send quantum states to any party it wishes to make payments to. The latter could be achieved, for example, in a network in which any two parties have the ability to request joint EPR pairs (i.e. maximally entangled pairs of qubits), which is one of the primary components of a ``quantum internet'' \cite{wehner2018quantum}. One can view our scheme as a possible use case of a quantum internet.
    \item The only known concrete constructions of a quantum lightning scheme rely on non-standard assumptions for security. The construction of Farhi et al. \cite{farhi2012quantum} contains no security proof. The construction of Zhandry \cite{zha19} is secure based on an assumption related to the multi-collision resistance of certain degree-2 hash functions. Arguably, neither of these assumptions is well-studied.
    \onote{Originally, it was: "Our construction is based on the existence of a secure quantum lightning scheme. There is only one known concrete construction for quantum lightning whose security hinges on a non-standard computational assumption about the multi-collision resistance of certain degree-2 hash functions." But what about the construction of Farhi et al.? See \cref{sec:farhi_et_al}. I changed it to the slightly more concise version above.}
\end{itemize}
In \cref{sec: comparison}, we extensively discuss the disadvantages (as well as the advantages) of our payment system vis-à-vis with other classical alternatives -- namely, a standard crypto-currency such as Bitcoin and second layer solutions such as the Lightning Network.

\paragraph{Outline} Section \ref{sec: prelim} covers preliminaries: \ref{sec: 2-1} covers basic notation; \ref{sec: lightning} introduces quantum lightning; \ref{sec: uc} gives a concise overview of the Universal Composability framework of Canetti \cite{canetti2001universally}. Section \ref{sec: blockchains} gives first an informal description of blockchains and smart contracts, followed by a formal definition of our global ideal functionality for a transaction ledger that handles stateful smart contracts. Section \ref{sec: main} describes our payment system. In Section \ref{sec: security}, we describe an adversarial model, and then prove security guarantees with respect to it. 

\section{Preliminaries} 
\label{sec: prelim}

\subsection{Notation}
\label{sec: 2-1}
For a function $f: \mathbb{N} \rightarrow \mathbb{R}$, we say that $f$ is \textit{negligible}, and we write $f(n) = negl(n)$, if for any positive polynomial $p(n)$ and all sufficiently large $n$'s, $f(n) < \frac{1}{p(n)}$. A \textit{binary} random variable is a random variable over $\{0,1\}$. We say that two ensembles of binary random variables $\{X_n\}$ and $\{Y_n\}$ are \textit{indistinguishable} if, 
$$\left| \, \Pr[X_n = 1] - \Pr[Y_n = 1]\,\right| = negl(n).$$ We use the terms PPT and QPT as abbreviations of probabilistic polynomial time and quantum polynomial time respectively.


\subsection{Quantum money and quantum lightning}
\label{sec: lightning}
Quantum money is a theoretical form of payment first proposed by Wiesner \cite{Wie83}, which replaces physical banknotes with quantum states. In essence, a quantum money scheme consists of a generation procedure, which mints banknotes, and a verification procedure, which verifies the validity of minted banknotes. A banknote consists of a quantum state together with an associated serial number. The appeal of quantum money comes primarily from a fundamental theorem in quantum theory, the No-Cloning theorem, which informally states that there does not exist a quantum operation that can clone arbitrary states. A second appealing property of quantum money, which is not celebrated nearly as much as the first, is that quantum money can be transferred almost instantaneously (by quantum teleportation for example). The first proposals for quantum money schemes required a central bank to carry out both the generation and verification procedures. The idea of public key quantum money was later formalized by Aaronson \cite{aaronson2009quantum}. In public-key quantum money, the verification procedure is public, meaning that anyone with access to a quantum banknote can verify its validity.

In this section, we focus on quantum lightning, a primitive recently proposed by Zhandry \cite{zha19},  and we enhance this to a decentralized quantum payment system. Informally, a quantum lightning scheme is a strengthening of public-key quantum money. It consists of a public generation procedure and a public verification procedure which satisfy the following two properties:
\begin{itemize}
\item Any quantum banknote generated by the honest generation procedure is accepted with probability negligibly close to $1$ by the verification procedure.
\item No adversarial generation procedure (not even the honest one) can generate two banknotes with the same serial number which both pass the verification procedure with non-negligible probability.
\end{itemize}

As mentioned earlier, there is only one known construction of quantum lightning, by Zhandry \cite{zha19}, who gives a construction which is secure under a computational assumption related to the multi-collision resistance of some degree-2 hash function. Zhandry also proves that any non-collapsing hash function can be used to construct quantum lightning. However, to the best of our knowledge, there are no known hash functions that are proven to be non-collapsing. In this section, we define quantum lightning formally, but we do not discuss any possible construction. Rather, in Section \ref{sec: main}, we will use quantum lightning as an off-the-shelf primitive.
\begin{definition}[Quantum lightning \cite{zha19}]
\label{def:quantum_lightning_completeness}
A quantum lightning scheme consists of a PPT algorithm $\textsf{QL.Setup}(1^\lambda)$ (where $\lambda$ is a security parameter) which samples a pair of polynomial-time quantum algorithms $(\textsf{gen-bolt}$, $\textsf{verify-bolt})$. $\textsf{gen-bolt}$ outputs pairs of the form $\left(\ket{\psi} \in \mathcal{H}_{\lambda}, s \in \{0,1\}^{\lambda}\right)$. We refer to $\ket{\psi}$ as a ``bolt'' and to $s$ as a ``serial number''. $\textsf{ver-bolt}$ takes as input a pair of the same form, and outputs either ``accept'' (1) or ``reject'' (0) (together with a post-measurement state). They satisfy the following:
\begin{itemize}
\item \begin{align}
    \Pr[\textsf{verify-bolt}(\ket{\psi},s) &=1: (\ket{\psi},s) \leftarrow \textsf{gen-bolt}, (\textsf{gen-bolt}, \textsf{verify-bolt}) \leftarrow \textsf{QL.Setup}] \\ &= 1-negl(\lambda)
\end{align}
\item For a state $\ket{\psi}$ let $\mathcal{M}_{\ket\psi}$ be the two outcome measurement which accepts $\ket{\psi}$ and rejects all states orthogonal to it.
\begin{align}
    \Pr[\mathcal M_{\ket{\psi}}(\ket{\psi'})=1&:(b,\ket{\psi'})\leftarrow \textsf{verify-bolt}(\ket{\psi},s), (\ket{\psi},s) \leftarrow \textsf{gen-bolt},  \\& (\textsf{gen-bolt}, \textsf{verify-bolt}) \leftarrow \textsf{QL.Setup}(\lambda)]= 1-negl(\lambda)
\end{align}
Here $\ket{\psi'}$ is the post-measurement state upon running $\textsf{verify-bolt}(\ket{\psi},s)$.
\item For all $s' \in \{0,1\}^{\lambda}$,
\begin{align}
\Pr[\textsf{verify-bolt}(\ket{\psi},s') =&1 \,\, \land \,\, s' \neq s: ((\ket{\psi},s) \leftarrow \textsf{gen-bolt}, (\textsf{gen-bolt}, \textsf{verify-bolt}) \leftarrow \textsf{QL.Setup}] \\&= negl(\lambda) 
\end{align}
\end{itemize}
\end{definition}
\mpar{\label{mar:no_perturbation}}
The three requirements simply ask that, with overwhelming probability, for any validly generated bolt $\ket{\psi}$, there is a single serial number $s$ such that $(\ket{\psi}, s)$ is accepted by the verification procedure, and that the verification does not perturb the bolt, except negligibly.

For security, we require that no adversarial generation procedure can produce two bolts with the same serial number. Formally, we define security via the following game between a challenger and an adversary $\mathcal{A}$.

\begin{itemize}
    \item The challenger runs $(\textsf{gen-bolt}, \textsf{verify-bolt}) \leftarrow \textsf{QL.Setup}(\lambda)$ and sends $(\textsf{gen-bolt}, \textsf{verify-bolt})$ to $\mathcal{A}$.
    \item $\mathcal{A}$ produces a pair $\ket{\Psi_{12}} \in \mathcal{H}_{\lambda}^{\otimes 2}, s \in \{0,1\}^{\lambda}$.
    \item The challenger runs $\textsf{verify-bolt}(\cdot,s)$ on each half of $\ket{\Psi_{12}}$. The output of the game is $1$ if both outcomes are ``accept''.
\end{itemize}

We let $\textsf{Counterfeit}(\lambda, \mathcal{A})$ be the random variable which denotes the output of the game.

\begin{definition}[Security \cite{zha19}]
\label{def: lightning basic security}
A quantum lightning scheme is secure if, for all polynomial-time quantum adversaries $\mathcal{A}$, $$ \Pr[\textsf{Counterfeit}(\lambda, \mathcal{A}) = 1 ] = \textnormal{negl}(\lambda).$$
\end{definition}

We define an additional property of a quantum lightning scheme, which in essence establishes that one can trade a quantum banknote for some useful classical certificate. Intuitively, this is meant to capture the fact that in the proposed construction of quantum lightning by Zhandry, one can measure a bolt with serial number $y$ in the computational basis to obtain a pre-image of $y$ under some hash function. However, doing so damages the bolt so that it will no longer pass verification. In order to define this additional property, we change the procedure $\textsf{QL.Setup}(1^{\lambda})$ slightly, so that it outputs a tuple $(\textsf{gen-bolt}, \textsf{verify-bolt}, \textsf{gen-certificate}, \textsf{verify-certificate})$, where $\textsf{gen-certificate}$ is a QPT algorithm that takes as input a quantum money state and a serial number and outputs a classical string of some fixed length $l(\lambda)$ for some polynomially bounded function $l$, which we refer to as a \textit{certificate}, and $\textsf{verify-certificate}$ is a PPT algorithm which takes as input a serial number and a certificate, and outputs ``accept'' ($1$) or ``reject'' ($0$). The additional property is defined based on the following game $\textsf{Forge-certificate}$ between a challenger and an adversary $\mathcal{A}$:

\begin{itemize}
    \item The challenger runs $(\textsf{gen-bolt}, \textsf{verify-bolt}, \textsf{gen-certificate}, \textsf{verify-certificate}) \leftarrow \textsf{QL.Setup}(1^\lambda)$ and sends the tuple to $\mathcal{A}$.
    \item $\mathcal{A}$ returns $c \in \{0,1\}^{l(\lambda)}$ and $(\ket{\psi}, s)$.
    \item The challenger runs $\textsf{verify-certificate}(s,c)$ and  $\textsf{verify-bolt}(\ket{\psi},s)$. Outputs $1$ if they both accept.
\end{itemize}

Let $\textsf{Forge-certificate}(\mathcal{A}, \lambda)$ be the random variable for the output of the challenger in the game above.

\begin{definition}[Trading the bolt for a classical certificate]
\label{def: extra property}
Let $\lambda \in \mathbb{N}$. We say that a quantum lightning scheme has ``bolt-to-certificate capability'' if:

\begin{itemize}
\item[(I)]
\begin{align}
\Pr[\textsf{verify-certificate}(s,c) =1: &(\textsf{gen-bolt}, \textsf{verify-bolt}, \textsf{gen-certificate}, \textsf{verify-certificate}) \leftarrow \textsf{QL.Setup}(1^{\lambda}), \\
&(\ket{\psi}, s) \leftarrow \textsf{gen-bolt}, \\ & c \leftarrow \textsf{gen-certificate}(\ket{\psi}, s)] = 1-negl(\lambda)
\end{align}
\item[(II)] For all polynomial-time quantum algorithms $\mathcal{A}$, $$ \Pr[ \textsf{Forge-certificate}(\mathcal{A}, \lambda)= 1 ] = negl(\lambda).$$
\end{itemize}
\end{definition}

Notice that property $(II)$ also implies that for most setups and serial numbers $s$ it is hard for any adversary to find a $c$ such that $\textsf{verify-certificate}(s,c) = 1$ without access to a valid state whose serial number is $s$. In fact, if there was an adversary $\mathcal{A}$ which succeeded at that, this could clearly be used to to construct an adversary $\mathcal{A'}$ that succeeds in $\textsf{Forge-certificate}$: Upon receiving $(\textsf{gen-bolt}, \textsf{verify-bolt}, \textsf{gen-certificate}, \textsf{verify-certificate})$ from the challenger, $\mathcal{A}'$ computes $(\ket{\psi},s) \leftarrow \textsf{gen-bolt}$; then runs $\mathcal{A}$ on input $s$ to obtain some $c$. $\mathcal{A}'$ returns $c$,$\ket{\psi}$ and $s$ to the challenger. We emphasize that in game $\textsf{Forge-certificate}$ it is the adversary himself who generates the state $\ket{\psi}$. This is important because when we employ the quantum lightning scheme later on parties are allowed to generate their own quantum banknotes.

\begin{prop}
\label{prop: extra property 1}
Any scheme that uses Zhandry's construction instantiated with a non-collapsing hash function \cite{zha19} satisfies the property of Definition \ref{def: extra property}. 
\end{prop}
\begin{proof}
Zhandry already proved that his scheme satisfies \cref{def:quantum_lightning_completeness,def: lightning basic security}. As we will see, since our construction does not change $\textsf{QL.Setup},\  \textsf{gen-bolt}$, and $ \textsf{verify-bolt}$, we only need to prove that it satisfies \cref{def: extra property}. We refer the reader to $\cite{zha19}$ for a definition of non-collapsing hash function.
In Zhandry's construction based on a non-collapsing hash function, $\textsf{QL.Setup}(1^{\lambda})$ outputs $(\textsf{gen-bolt}, \textsf{verify-bolt})$. A bolt generated from \textsf{gen-bolt} has the form $\ket{\Psi} = \bigotimes_{i=1}^n \ket{\psi_{y_i}}$, where $y_i \in \{0,1\}^{\lambda}$ for all $i$ and $\ket{\psi_{y_i}} = \sum_{x: H(x) = y_i} \ket{x}$, where $H$ is a non-collapsing hash-function, and $n \in \mathbb{N}$ is polynomial in the security parameter $\lambda$. \textsf{verify-bolt} has the form of $n$-fold product of a verification procedure ``$\textnormal{Mini-Ver}$'' which acts on a single one of the $n$ registers, though it is not crucial to understand how $\textnormal{Mini-Ver}$ for this work. 
In Zhandry's construction, the serial number associated to the bolt $\ket{\Psi}$ above is $s = (y_1,\ldots,y_n)$.

We define $\gencertificate$ as the QPT algorithm that measures $\ket{\psi}$ in the standard basis, and outputs the outcome. When applied to an honestly generated bolt, the outcomes are pre-images $x_i$ of $y_i$, for $i=1,..,n$. We define $\verifycertificate$ as the deterministic algorithm which receives a serial number $s=(y_1,\ldots,y_n)$ and a certificate $c=(x_1,\ldots,x_n)$ and checks that for all $i\in [n]$, $H(x_i)=y_i$.

Is is clear that $(I)$ holds. 

For property $(II)$, suppose there exists $\mathcal{A}$ such that $\Pr[\textsf{Forge-certificate}(\mathcal{A}, \lambda)= 1]$ is non-negligible. We use $\mathcal{A}$ to construct an adversary $\mathcal{A}'$ that breaks collision-resistance of a non-collapsing hash function as follows: $\mathcal{A}'$ runs $(\textsf{gen-bolt}, \textsf{verify-bolt}) \leftarrow \textsf{QL.Setup}(1^\lambda)$. Let $H$ be the non-collapsing hash function hard-coded in the description of \textsf{verify-bolt}. $\mathcal{A'}$ sends the tuple $(\textsf{gen-bolt}, \textsf{verify-bolt}, \textsf{gen-certificate}, \textsf{verify-certificate})$ to $\mathcal{A}$, where the latter two are defined as above in terms of $H$. $\mathcal{A}$ returns $(c, \ket{\psi})$, where $c$ is parsed as $c = (x_1,..,x_n)$.
 $\mathcal{A}'$ then measures each of the $n$ registers of $\ket{\psi}$ to get $c' = (x'_1,..,x'_n)$. If $x_i \neq x_i'$, then $\mathcal{A}'$ outputs $(x_i, x_i')$. We claim that with non-neglibile probability $\mathcal{A'}$ outputs a collision for $H$. To see this, notice that since $\mathcal{A}$ wins $\textsf{Forge-certificate}$ with non-negligible probability, then $\ket{\psi}$ must pass \textsf{verify-bolt} with non-negligible probability; and from the analysis of \cite{zha19}, any such state must be such that at least one of the registers is in a superposition which has non-negligible weight on at least two pre-images. For more details, see the proof of Theorem 5 in \cite{zha19}, and specifically the following claim:\mpar{\label{mar:two_preimages}}
 \begin{quote}
     If the bolts are measured, two different pre-images of the same y, and hence
a collision for $H^{\tensor r}$, will be obtained with probability at least 1/200.
 \end{quote}
\end{proof}

\begin{prop}
\label{prop: extra property 2}
Zhandry's construction based on the multi-collision resistance of certain degree-2 hash functions (from section 6 of \cite{zha19}) satisfies the property of Definition \ref{def: extra property}. 
\end{prop}
The proof is similar to the proof of Proposition \ref{prop: extra property 1}. We include it for completeness in Appendix \ref{sec: appendix}. 

\begin{fact}
Zhandry's construction from \cref{prop: extra property 2} only requires a common random string (CRS) for $\textsf{QL.Setup}$.
\label{fac: Zhandry requires CRS}
\end{fact}
This fact will be relevant when we discuss a concrete implementation in a system such as Bitcoin. For more details, see \cref{sec: practical implementation on Bitcoin}.
\paragraph{Other Quantum Lightning Schemes?}
\label{par:farhi_et_al}
Farhi et al.~\cite{FGH+12} constructed a public quantum money scheme which they conjectured to have the following property:
\begin{quote}A rogue mint running the same algorithm as the mint can produce a
new set of money pairs, but (with high probability) none of the serial numbers
will match those that the mint originally produced.
\end{quote}
This is less formal, but we believe captures Zhandry's security definition of a quantum lightning (though it was introduced roughly 7 years earlier). 

Very recently, Peter Shor gave a talk on ongoing (unpublished) work for a lattice-based quantum lightning scheme (see \url{https://youtu.be/8fzLByTn8Xk}).

Here, we briefly compare the two existing published constructions, and we omit Shor's scheme from this discussion. 
\begin{itemize}
    \item Farhi et al. only provide several plausibility arguments supporting the security of their scheme, but their work does not contain any security proof. In a follow-up work, Lutomirski~\cite{Lut11} showed a security reduction from a problem \emph{related} to counterfeiting Farhi et al.'s money scheme.  Even though it does not have a (proper) security proof, Farhi et al.'s scheme was first published 8 years ago, and was not broken since then. 
Zhandry's construction is proved to be secure under a non-standard hardness assumption, which was first introduced in that work. In his Eurocrypt talk, Zhandry reports that the construction is ``broken in some settings'' -- see \url{https://youtu.be/fjumbNTZSik?t=1302}.\footnote{Additionally, the hardness assumption was changed between the second and third version on the IACR e-print. This might suggest that the original hardness assumption was also broken. The most updated e-print version does not discuss that discrepancy. }
\item The scheme of Farhi et al. does not require a trusted setup. 
\item The scheme of Farhi et al. does not have a bolt-to-certificate capability. This property is required in our payment system to transform quantum banknotes back to coins on the blockchain -- see \cref{sec: main}.
\end{itemize}

\subsection{Universal Composability}
\label{sec: uc}
This section is intended as a concise primer about the Universal Composability (UC) model of Canetti \cite{canetti2001universally}. We refer the reader to \cite{canetti2001universally} for a rigorous definition and treatment of the UC model and of composable security, and to the tutorial $\cite{canetti2006tutorial}$ for a more gentle introduction. At the end, we provide a brief overview of the Generalized UC model (GUC) \cite{canetti2007global}. While introducing UC and GUC, we also setup some of the notation that we will employ in the rest of the paper. As mentioned in the Introduction, we elect to work in the GUC, because this provides a convenient formal abstraction layer for modeling a ledger functionality that supports smart-contracts, and allows for a clean security proof in the idealized setting in which parties have access to this functionality.

The reader familiar with the UC and GUC framework may wish to skip ahead to the next section.

In the universal composability framework (UC), parties are modelled as Interactive Turing Machines (ITM) who can communicate by writing on each other's externally writable tapes, subject to some global constraints.  Informally, a protocol is specified by a code $\pi$ for an ITM, and consists of various rounds of communication and local computation between instances of  ITMs (the parties), each running the code $\pi$ on their machine, on some private input.

Security in the UC model is defined via the notion of emulation. Informally, we say that a protocol $\pi$ emulates a protocol $\phi$ if whatever can be achieved by an adversary attacking $\pi$ can also be achieved by some other adversary attacking $\phi$. This is formalized by introducing simulators and environments.

Given protocols $\pi$ and $\phi$, we say that $\pi$ emulates (or "is as secure as") $\phi$ in the UC model, if for any polynomial-time adversary $\mathcal{A}$ attacking protocol $\pi$, there exists a polynomial-time simulator $\mathcal{S}$ attacking $\phi$ such that no polynomial-time distinguisher $\mathcal{E}$, referred to as the \textit{environment}, can distinguish between $\pi$ running with $\mathcal{A}$ and $\phi$ running with $\mathcal{S}$. Here, the environment $\mathcal{E}$ is allowed to choose the protocol inputs, read the protocol outputs, including outputs from the adversary or the simulator, and to communicate with the adversary or simulator during the execution of the protocol (without of course being told whether the interaction is with the adversary or with the simulator). In this framework, one can formulate security of a multiparty cryptographic task by first defining an \textit{ideal functionality} $\mathcal{F}$ that behaves exactly as intended, and then providing a ``real-world'' protocol that emulates, or ``securely realizes'', the ideal functionality $\mathcal{F}$. 

We give more formal definitions for the above intuition. To formulate precisely what it means for an environment $\mathcal{E}$ to tell two executions apart, one has to formalize the interaction between $\mathcal{E}$ and the protocols in these executions. Concisely, an execution of a protocol $\pi$ with adversary $\mathcal{A}$ and environment $\mathcal{E}$ consists of a sequence of activations of ITMs. At each activation, the active ITM runs according to its code, its state and the content of its tapes, until it reaches a special \textsf{wait} state. The sequence of activations proceeds as follows: The environment $\mathcal{E}$ gets activated first and chooses inputs for $\mathcal{A}$ and for all parties. Once $\mathcal{A}$ or a party is actived by an incoming message or an input, it runs its code until it produces an outgoing message for another party, an output for $\mathcal{E}$, or it reaches the $\textsf{wait}$ state, in which case $\mathcal{E}$ is activated again. The execution terminates when $\mathcal{E}$ produces its output, which can be taken to be a single bit. Note that each time it is activated, $\mathcal{E}$ is also allowed to invoke a new party, and assign a unique PID (party identifier) to it. 
Allowing the environment to invoke new parties will be particularly important in Section \ref{sec: security}, where we discuss security. There, the fact that the environment has this ability implies that our security notion captures realistic scenarios in which the set of parties is not fixed at the start, but is allowed to change. Moreover, each invocation of a protocol $\pi$ is assigned a unique session identifier SID, to distinguish it from other invocations of $\pi$. We denote by $\textrm{EXEC}_{\pi,\mathcal{A}, \mathcal{E}}(\lambda, z)$ the output of environment $\mathcal{E}$ initialized with input $z$, and security parameter $\lambda$ in an execution of $\pi$ with adversary $\mathcal{A}$.  

We are ready to state the following (slightly informal) definition.

\begin{definition}
\label{def: emulation}
A protocol $\pi$ UC-emulates a protocol $\phi$ if, for any PPT adversary $\mathcal{A}$, there exists a PPT simulator $\mathcal{S}$ such that, for any PPT environment $\mathcal{E}$, the families of random variables $\{\textrm{EXEC}_{\pi,\mathcal{A}, \mathcal{E}}(\lambda, z)\}_{\lambda \in \mathbb{N}, z \in \{0,1\}^{poly(\lambda)}}$ and $\{\text{EXEC}_{\phi,\mathcal{S}, \mathcal{E}}(\lambda, z)\}_{\lambda \in \mathbb{N}, z \in \{0,1\}^{poly(\lambda)}}$ are indistinguishable.
\end{definition}

Then, given an ideal functionality $\mathcal{F}$ which captures the intended ideal behaviour of a certain cryptographic task, one can define the ITM code $I_{\mathcal{F}}$, which behaves as follows: the ITM running $I_\mathcal{F}$ simply forwards any inputs received to the ideal functionality $\mathcal{F}$. We then say that a ``real-world'' protocol $\pi$ securely realizes $\mathcal{F}$ if $\pi$ emulates $I_{\mathcal{F}}$ according to Definition \ref{def: emulation}.

\paragraph{A composition theorem} The notion of security we just defined is strong. One of the main advantages of such a security definition is that it supports composition, i.e. security remains when secure protocols are executed concurrently, and arbitrary messages can be sent between executions. We use the notation $\sigma^{\pi}$ for a protocol $\sigma$ that makes up to polynomially many calls to another protocol $\pi$. In a typical scenario, $\sigma^{\mathcal{F}}$ is a protocol that makes use of an ideal functionality $\mathcal{F}$, and $\sigma^{\pi}$ is the protocol that results by implementing $\mathcal{F}$ through the protocol $\pi$ (i.e. replacing calls to $\mathcal{F}$ by calls to $\pi$). It is natural to expect that if $\pi$ securely realizes $\mathcal{F}$, then $\sigma^{\pi}$ securely realizes $\sigma^{\mathcal{F}}$. This is the content of the following theorem.

\begin{theorem}[Universal Composition Theorem]
Let $\pi$, $\phi$, $\sigma$ be polynomial-time protocols. Suppose protocol $\pi$ UC-emulates $\phi$. Then $\sigma^{\pi}$ UC-emulates $\sigma^{\phi}$.
\end{theorem}

Replacing $\phi$ by $I_{\mathcal{F}}$ for some ideal functionality $\mathcal{F}$ in the above theorem yields the composable security notion discussed above.


\paragraph{Generalized UC model}
The formalism of the original UC model is not able to handle security requirements in the presence of a ``global trusted setup''. By this, we mean some global information accessible to all parties, which is guaranteed to have certain properties.  Examples of this are a public-key infrastructure or a common reference string. Emulation in the original UC sense is not enough to guarantee composability properties in the presence of a global setup. Indeed, one can construct examples in which a UC-secure protocol for some functionality interacts badly with another UC-secure protocol and affects its security, if both protocols make reference to the same global setup. For more details and concrete examples see \cite{canetti2007global}.

The generalized UC framework (GUC) of Canetti et al. \cite{canetti2007global} allows for a ``global setup''. The latter is modelled as an ideal functionality which is allowed to interact not only with the parties running the protocol, but also with the environment. GUC formulates a stronger security notion, which is sufficient to guarantee a composition theorem, i.e. ideal functionalities with access to a shared global functionality $\mathcal{G}$ can be replaced by protocols that securely realize them in the presence of $\mathcal{G}$. Further, one can also replace global ideal functionalities with appropriate protocols realizing them. This kind of replacement does not immediately follow from the previous composition theorem and requires a more careful analysis, as is done in \cite{canetti2016pki}, where sufficient conditions for this replacement are established.

\paragraph{Universal Composability in the quantum setting}
In our setting, we are interested in honest parties, adversaries and environments that are quantum polynomial-time ITMs. The notion of Universal Composability has been studied in the quantum setting in \cite{benor2004general}, \cite{unruh2004simulatable} and \cite{unruh2010universally}. In particular, in \cite{unruh2010universally}, Unruh extends the model of computation of UC and its composition theorems to the setting in which polynomial-time classical ITMs are replaced by polynomial-time quantum ITMs (and ideal functionalities are still classical). The proofs are essentially the same as in the classical setting. Although the quantum version of the Generalized UC framework has not been explicitly studied in \cite{unruh2010universally}, one can check that the proofs of the composition theorems for GUC from \cite{canetti2007global} and \cite{canetti2016pki} also go through virtually unchanged in the quantum setting.

\section{Blockchains and smart contracts}
\label{sec: blockchains}
In this section, we start by describing blockchains and smart contracts informally. We follow this by a more formal description. As mentioned in the introduction, the essential features of blockchains and smart contracts can be abstracted by modeling them as ideal functionalities in the Universal Composability framework of Canetti \cite{canetti2001universally}. In this section, we introduce a global ideal functionality that abstracts the properties of a transaction ledger capable of handling stateful smart contracts. We call this $\mathcal{F}_{Ledg}$, which we describe in Fig. \ref{fig: ledger functionality}. \anote{The rest of this paragraph is new} We remark that $\mathcal{F}_{Ledg}$ is somewhat of an idealized functionality which allows us to obtain a clean description and security proof for our payment system. However,  $\mathcal{F}_{Ledg}$ does not capture, for example, attacks that stem from miners possibly delaying honest parties' messages from being recorded on the blockchain. We discuss this and other issues extensively in Section \ref{sec: practical issues}. We also discuss ways to resolve such issues in detail, but we do not formally incorporate these in the description of our our payment system in Section \ref{sec: main}: since we view our contribution as primarily conceptual, we elect to keep the exposition and security proof of the basic payment system as clean and accessible as possible. 
\vspace{1.5mm}

Informally, a blockchain is a public ledger consisting of a sequence of blocks. Each block typically contains information about a set of transactions, and a new block is appended regularly via a consensus mechanism that involves the nodes of a network. A blockchain is equipped with a native currency which is employed in transactions, and whose basic unit is referred to as a \textit{coin}.

Each user in the network is associated with a public key (this can be thought of as the user's address). A typical transaction is a message which transfers coins from a public key to another. It is considered valid if it is digitally signed using the secret key corresponding to the sending address. 

More precisely, in Bitcoin, parties do not keep track of users's accounts, but rather they just maintain a local copy of a set known as ``unspent transaction outputs set'' (UTXO set). An unspent output is a transaction that has not yet been ``claimed'', i.e. the coins of these transactions have not yet been spent by the receiver. Each unspent output in the UTXO set includes a circuit (also known as a ``script'') such that any user that can provide an input which is accepted by the circuit (i.e. a witness) can make a transaction that spends these coins, thus creating a new unspent output. Hence, if only one user knows the witness to the circuit, he is effectively the owner of these coins. For a standard payment transaction, the witness is a signature and the circuit verifies the signature. However, more complex circuits are also allowed, and these give rise to more complex transactions than simple payments: smart contracts. A smart contract can be thought of as a transaction which deposits coins to an address. The coins are released upon fulfillment of certain pre-established conditions. 

In \cite{bentov2014use} and \cite{bentov2017instantaneous}, smart contracts are defined as ideal functionalities in a variant of the Universal Composability (UC) model \cite{canetti2001universally}. The ideal functionality that abstracts the simplest smart contracts was formalized in \cite{bentov2014use}, and called ``Claim or Refund''. Informally, this functionality specifies that a sender $P$ locks his coins and chooses a circuit $\phi$, such that a receiver $Q$ can gain possession of these coins by providing a witness $w$ such that $\phi(w) = 1$ before an established time, and otherwise the sender can reclaim his coins. The ``Claim or Refund'' ideal functionality can be realized in Bitcoin as long as the circuit $\phi$ can be described in Bitcoin's scripting language. On the other hand, Ethereum's scripting language is Turing-complete, and so any circuit $\phi$ can be described.

``Claim or refund'' ideal functionalities can be further generalized to ``stateful contracts''. In Ethereum, each unspent output also maintains a \textit{state}. In other words, each unspent output comprises not only a circuit $\phi$, but also state variables. Parties can claim partial amounts of coins by providing witnesses that satisfy the circuit $\phi$ in accordance with the current state variables. In addition, $\phi$ also specifies an update rule for the state variables, which are updated accordingly. We refer to these type of transactions as stateful contracts, as opposed to the ``stateless'' contract of ``Claim or Refund''. 
Stateful contracts can be realized in Ethereum, but not in Bitcoin. From now onwards, we will only work with stateful contracts. We will use the terms ``smart contracts'' and ``stateful contracts'' interchangeably. 

We emphasize that our modeling is inspired by \cite{bentov2014use} and \cite{bentov2017instantaneous}, but differs in the way that coins are formalized. One difference from the model of Bentov et al. is that there, in order to handle coins, the authors augment the original UC model by endowing each party with a \textit{wallet} and a \textit{safe}, and by considering \textit{coins} as atomic entities which can be exchanged between parties. To the best of our knowledge, this variant of the UC framework is not subsumed by any of the previously studied variants, and thus it is not known whether a composition theorem holds for it.

On the other hand, we feel that a more natural approach is to work in the Generalized UC model \cite{canetti2007global}, and to define a global ideal functionality $\mathcal{F}_{Ledg}$ which abstracts the essential features of a \textit{transaction ledger} capable of handling stateful smart contracts. This approach was first proposed by Kiayias et al. \cite{kiayias2016fair}. The appeal of modeling a transaction ledger as a global ideal functionality is that composition theorems are known in the Generalized UC framework. In virtue of this, any secure protocol for some task that makes calls to $\mathcal{F}_{Ledg}$ can be arbitrarily composed while still maintaining security. This means that one need not worry about composing different concurrent protocols which reference the same transaction ledger. One would hope that it is also the case that a secure protocol that makes calls to $\mathcal{F}_{Ledg}$ remains secure when the latter are replaced by calls to secure real-world realizations of it (on Ethereum for example). This requires a more careful analysis, and Canetti et al. provide in \cite{canetti2016pki} sufficient conditions for this replacement to be possible. We do not prove that a secure real-world realization of $\mathcal{F}_{Ledg}$ on an existing blockchain exists, but we believe that $\mathcal{F}_{Ledg}$, or a close variant of it, should be securely realizable on the Ethereum blockchain. In any case, we work abstractly by designing our payment system, and proving it secure, assuming access to such an ideal functionality. The appeal of such an approach is that the security of the higher-level protocols is independent of the details of the particular real-world implementation of $\mathcal{F}_{Ledg}$. 

Next, we describe our global ideal functionality $\mathcal{F}_{Ledg}$. In doing so, we establish the notation that we will utilize in the rest of the paper. 

\vspace{1.5mm}

\paragraph{Global ledger ideal functionality} 


We present in Fig. \ref{fig: ledger functionality} our global ledger ideal functionality $\mathcal{F}_{Ledg}$. In a nutshell, this keeps track of every registered party's coins, and allows any party to transfer coins in their name to any other party. It also allows any party to retrieve information about the number of coins of any other party, as well as about any previous transaction. The initial amount of coins of a newly registered party is determined by its PID (recall that in a UC-execution the PID of each invoked party is specified by the environment; see Section \ref{sec: uc} for more details). Moreover, $\mathcal{F}_{Ledg}$ handles (stateful) smart contracts: it accepts deposits from the parties involved in a contract and then pays rewards appropriately. Recall that in stateful smart contracts a party or a set of parties deposit an amount of coins to the contract. The contract is specified by a circuit $\phi$, together with an initial value for a state variable $\textsf{st}$. A state transition is triggered by any party $P$ with PID $pid$ sending a witness $w$ which is accepted by $\phi$ in accordance with the current state and the current time $t$. More precisely, the contract runs $\phi(pid,w,t,\textsf{st})$, which outputs either ``$\perp$'' or a new state (stored in the variable \textsf{st}) and a number of coins $d \in \mathbb{N}$ that is released to $P$. Each contract then repeatedly accepts state transitions until it has distributed all the coins that were deposited into it at the start. Notice that $\phi$ can accept different witnesses at different times (the acceptance of a witness can depend on the current time $t$ and the current value of the state variable $\textsf{st}$). Information about existing smart contracts can also be retrieved by any party.

The stateful-contract portion of $\mathcal{F}_{Ledg}$ resembles closely the functionality $\mathcal{F}_{StCon}$ from \cite{bentov2017instantaneous}. Our approach differs from that of $\cite{bentov2017instantaneous}$ in that the we make the smart contract functionality part of the global ideal functionality $\mathcal{F}_{Ledg}$ which also keeps track of party's coins and transactions. In \cite{bentov2017instantaneous} instead, coins are incorporated in the computation model by augmenting the ITMs with wallets and safes (this changes the model of computation in a way that is not captured by any of the the previously studied variants of the UC framework).

We implicitly assume access to an ideal functionality for message authentication $\mathcal{F}_{auth}$ which all parties employ when sending their messages, and also to a global ideal functionality for a clock that keeps track of time. We assume implicitly that $\mathcal{F}_{Ledg}$ makes calls to the clock and keeps track of time. Alternatively, we could just have $F_{Ledg}$ maintain a local variable that counts the number of transactions performed, and a local time variable $t$, which is increased by $1$ every time the number of transactions reaches a certain number, after which the transaction counter is reset (this mimics the process of addition of blocks in a blockchain, and time simply counts the number of blocks). From now onwards, we do not formally reference calls to $\mathcal{F}_{auth}$ or to the clock to avoid overloading notation. We are now ready to define $\mathcal{F}_{Ledg}$.

\begin{figure}[H]
\rule[1ex]{16.5cm}{0.5pt}
{\centering \textbf{Global ledger ideal functionality} \par}

Initialize the sets $\textsf{parties} = \{\}$, $\textsf{contracts} = \{\}$, $\textsf{AllTransactions} = \{\}$. (Throughout, the variable $t$ denotes the current time.)

\paragraph{Add party:} Upon receiving a message \textsf{AddParty} from a party with PID $pid = (id, d)$, send (\textsf{AddedParty}, $pid$) to the adversary; upon receiving a message \textsf{ok} from the adversary, and if this is the first request from $pid$, add $pid$ to the set $\textsf{parties}$. Set $pid.\textsf{id} \leftarrow id$ and $pid.\textsf{coins} \leftarrow d$.

\paragraph{Retrieve party:} Upon receiving a message (\textsf{RetrieveParty}, $pid$) from some party $P$ (or the adversary), output (\textsf{RetrieveParty}, $pid$, $d$) to $P$ (or to the adversary), where $d = \perp$ if $pid \notin \textsf{parties}$, else $d = pid.\textsf{coins}$. (We abuse notation here in that, when taken as part of a message, $pid$ is treated as a string, but, when called by the functionality, it is a variable with attributes $pid.\textsf{id}$ and $pid.\textsf{coins}$).

\paragraph{Add transaction:} Upon receiving a message ($\textsf{AddTransaction}$, $pid'$, $d$) from some party $P$ with PID $pid$, and $pid \in \textsf{parties}$, do the following:
\begin{itemize}
    \item If $pid' \in \textsf{parties}$ and $pid.\textsf{coins} >d$, update $pid'.\textsf{coins} \leftarrow pid'.\textsf{coins} + d$ and $pid.\textsf{coins}  \leftarrow pid.\textsf{coins} - d$. Set $trId = |\textsf{AllTransactions}| + 1$. Add a variable named $trId$ to $\textsf{AllTransactions}$, with attribute $trId.\textsf{transaction} = (pid, pid', d, t)$. Send a message (\textsf{Executed}, $trId$) to $P$.
    \item Else, return $\perp$ to $P$.
\end{itemize}

\paragraph{Retrieve transaction:} Upon receiving a message (\textsf{RetrieveTransaction}, $trId$) from some party $P$ (or the adversary), output (\textsf{RetrieveTransaction}, $trId$, $s$), where $s = \perp$ if $trId \notin \textsf{allTransactions}$, and $s = trId.\textsf{transaction}$ otherwise.

\paragraph{Add/trigger smart contract:} Upon receiving a message (\textsf{AddSmartContract}, \textsf{Params}=$(I, D, \phi, \textsf{st}_0)$), where $I$ is a set of PID's, $D$ is a set $\{(pid, d_{pid}): pid \in I\}$ of ``initial deposits'', with $d_{pid}$ being the amount required initially from the party with PID $pid$, $\phi$ is a circuit, and $\textsf{st}_0$ is the initial value of a state variable $\textsf{st}$, check that $I \subseteq \textsf{parties}$. If not, ignore the message; if yes, set $ssid = |\textsf{contracts}| + 1$. Add a variable named $ssid$ to $\textsf{contracts}$ with attributes $ss id.\textsf{Params} = (I, D, \phi, \textsf{st}_0)$, $ssid.\textsf{state} = \textsf{st}$ and $ssid.\textsf{coins} \leftarrow 0$. Send a message (\textsf{RecordedContract}, $ssid$) to $P$. Then, do the following:
\begin{itemize}
    \item \textbf{Initialization phase:} Wait to get message $(\textsf{InitializeWithCoins}, \textit{ssid}, \textsf{Params}=(I, D, \phi, \textsf{st}_0))$ from party with PID $pid$ for all $pid \in I$. When all messages are received, and if, for all $pid \in I$, $pid.\textsf{coins} \geq d_{pid}$, then, for all $pid \in I$, update: $pid.\textsf{coins} \leftarrow pid.\textsf{coins} - d_{pid}$ and $ssid.\textsf{coins} \leftarrow ssid.\textsf{coins} + d_{pid}$. Set $\textsf{st} \leftarrow \textsf{st}_0$  (We assume that $ssid.\textsf{state}$ changes dynamically with $\textsf{st}$).
    \item \textbf{Execution phase:} Repeat until termination: Upon receiving a message of the form $(\textsf{Trigger}, \textit{ssid}, w, d)$ at time $t$ from some party with PID $pid \in \textsf{parties}$ (where it can also be $d=0$) such that $\phi(pid, w, t, \textsf{st}, d) \neq \perp$, do the following:
    \begin{itemize}
        \item If $d>0$, update $pid.\textsf{coins} \leftarrow pid.\textsf{coins} - d$ and $ssid.\textsf{coins} \leftarrow ssid.\textsf{coins} + d$
        \item Update $(\textsf{st}, e) \leftarrow \phi(pid, w, t, \textsf{st}, d)$. 
        \item If $e = \textnormal{``all coins''}$, let $q := ssid.\textsf{coins}$. Send the message (\textsf{Reward}, $ssid$, $q$) to the party with PID $pid$ and update $pid.\textsf{coins} \leftarrow pid.\textsf{coins} + q$ and $ssid.\textsf{coins} \leftarrow 0$. If $e > 0$ and $ssid.\textsf{coins} > e$, send the message (\textsf{Reward}, $ssid$, $e$) to the party with PID $pid$, and update $pid.\textsf{coins} \leftarrow pid.\textsf{coins} + e$ and $ssid.\textsf{coins} \leftarrow ssid.\textsf{coins} -e$. Else, if $ssid.\textsf{coins} = e' \leq e$, send the message (\textsf{Reward}, $ssid$, $e'$) to the party with PID $pid$, and update $pid.\textsf{coins} \leftarrow pid.\textsf{coins} + e'$ and $ssid.\textsf{coins} \leftarrow 0$. Then, terminate.
    \end{itemize}
\end{itemize}

\paragraph{Retrieve smart contract:} Upon receiving a message (\textsf{RetrieveContract}, $ssid$) from some party $P$ (or the adversary), output (\textsf{RetrieveContract}, $ssid$, $z$), where $z = \perp$ if $ssid \notin \textsf{contracts}$, and $z = (ssid.\textsf{Params}, ssid.\textsf{state}, ssid.\textsf{coins})$ otherwise. 

\rule[2ex]{16.5cm}{0.5pt}\vspace{-.5cm}
\caption{Global ledger ideal functionality $\mathcal{F}_{Ledg}$}
  \label{fig: ledger functionality}
 
\end{figure}

We think of "Retrieve" operations as being fast, or free, as they do not alter the state of the ledger. \mpar{\label{mar:fast_vs_slow}}We think of "Add" and "Trigger" operations as being slow, as they alter the state of the ledger.

From now on, we will often refer to the number of coins $ssid.\textsf{coins}$ of a contract with session identifier $ssid$ as the coins \textit{deposited} in the contract. When we say that a contract $\textit{releases}$ some coins to a party $P$ with PID $pid$, we mean more precisely that $\mathcal{F}_{Ledg}$ updates its local variables and moves coins from $ssid.\textsf{coins}$ to $pid.\textsf{coins}$.

\section{A payment system based on quantum lightning and a classical blockchain}
\label{sec: main}
In this section, we describe our payment system. We give first an informal description, and in Section \ref{sec: 3-2} we give a formal description. 

The building block of our payment system is a quantum lightning scheme, reviewed in detail in Section \ref{sec: lightning}. Recall that a quantum lightning scheme consists of a generation procedure which creates quantum banknotes, and a verification procedure that verifies them and assigns serial numbers. The security guarantee is that no generation procedure (not even the honest one) can create two banknotes with the same serial number except with negligible probability. As mentioned earlier, this property is desirable if one wants to design a decentralized payment system, as it prevents anyone from cloning banknotes (even the person who generates them). However, this calls for a mechanism to regulate generation of new valid quantum banknotes.

In this section, we describe formally a proposal that employs smart contracts to provide such a mechanism. As we have described informally in the introduction, the high-level idea is to keep track of the valid serial numbers using smart contracts. Any party is allowed to deposit any amount of coins $d$ (of their choice) into a smart contract with specific parameters (see definition \ref{def: banknote-contract} below), and with an initial value of his choice for a $\textit{serial number}$ state variable. We can think of the quantum banknote with the chosen serial number as having ``acquired'' value $d$. A payment involves only two parties: a payer, who sends a quantum banknote, and a payee who receives it and verifies it locally. As anticipated in the introduction, the full payment system includes the following additional features, which we describe here informally in a little more detail (all of these are described formally in Section \ref{sec: 3-2}):

\begin{itemize}
\item Removing a serial number from the list of valid serial numbers in order to recover the amount of coins deposited in the corresponding smart contract. This makes the two commodities (quantum banknotes and coins on the blockchain) interchangeable. This is achieved by exploiting the additional property of of some quantum lightning scheme from Definition \ref{def: extra property}. Recall that, informally, this property states that there is some classical certificate that can be recovered by measuring a valid quantum banknote, which no efficient algorithm can recover otherwise. The key is that once the state is measured to recover this certificate, it is damaged in a way that it only passes verification with negligible probability (meaning that it can no longer be spent). We allow users to submit this classical certificate to a smart contract, and if the certificate is consistent with the serial number stored in the contract, then the latter releases all of the coins deposited in the contract to the user. 
\item Allowing a party to replace an existing serial number with a new one of their choice in case they lose a valid quantum state (they are fragile after all!). We allow a user $P$ to file a ``lost banknote claim'' by sending a message and some fixed amount of coins $d_0$ to a smart contract whose serial number is the serial number of the lost banknote. The idea is that if no one challenges this claim, then after a specified time $t_{tr}$ user $P$ can submit a message which changes the value of the serial number state variable to a value of his choice and recovers the previously deposited $d_0$ coins. On the other hand, if a user $P$ maliciously files a claim to some contract with serial number $s$, then any user $Q$ who possesses the valid banknote with serial number $s$ can recover the classical certificate from Definition \ref{def: extra property}, and submit it to the contract. This releases all the coins deposited in the contract to $Q$ (including the $d_0$ deposited by $P$ to make the claim). As you might notice, this requires honest users to monitor existing contracts for ``lost banknote claims''. This, however, is not much of a burden if $t_{tr}$ is made large enough (say a week or a month). The requirement of being online once a week or once a month is easy to meet in practice.
\end{itemize}

\subsection{The payment system and its components}
\label{sec: 3-2}

In this section, we describe in detail all of the components of the payment system. It consists of the following: a protocol to generate valid quantum banknotes; a protocol to make a payment; a protocol to file a claim for a lost banknote; a protocol to prevent malicious attempts at filing claims for lost banknotes; and a protocol to trade a valid quantum banknote in exchange for coins.

Let $\lambda \in \mathbb{N}$. From now onwards, we assume that\\ $$(\textsf{gen-bolt}, \textsf{verify-bolt}, \textsf{gen-certificate}, \textsf{verify-certificate}) \leftarrow \textsf{QL.Setup}(\lambda),$$ where the latter is the setup procedure of a quantum lightning scheme with bolt-to-certificate capability (i.e. a quantum lightning scheme that satisfies the additional property of Definition \ref{def: extra property}). We are thus assuming a trusted setup for our payment system (note there are practical ways to ensure that such a setup, which is a one-time procedure, is performed legitimately even in a network where many parties might be dishonest). 
We also assume that all parties have access to an authenticated and ideal quantum channel.\mpar{\label{mar:ideal_authenticated}}

Recall from the description of $\mathcal{F}_{Ledg}$ that each smart contract is specified by several parameters: $I$ is a set of PIDs of parties who are expected to make the initial deposit, with $\{d_{pid} : pid \in I\}$ being the required initial deposit amounts; a circuit $\phi$ specifies how the state variables are updated and when coins are released; an initial value $\textsf{st}_0$ for the state variable $\textsf{st}$; a session identifier $ssid$.

In Definition \ref{def: banknote-contract} below, we define an instantiation of smart contracts with a particular choice of parameters, which we refer to as banknote-contracts. Banknote-contracts are the building blocks of the protocols that make up our payment system. We describe a banknote-contract informally before giving a formal definition. 

A banknote-contract is a smart contract initialized by a single party, and it has a state variable of the form $\textsf{st} = (\textsf{serial}, \textsf{ActiveLostClaim})$. The party initializes the banknote-contract by depositing a number of coins $d$ and by setting the initial value of $\textsf{serial}$ to any desired value. The banknote-contract handles the following type of requests:
\begin{itemize}
    \item As long as $\textsf{ActiveLostClaim} = \text{``No active claim''}$ (which signifies that there are no currently active lost-banknote claims), any party $P$ can send the message \textsf{BanknoteLost}, together with a pre-established amount of coins $d_0$ to the contract. This will trigger an update of the state variable $\textsf{ActiveLostClaim}$ to reflect the active lost-banknote claim by party $P$. 
    \item As long as there is an active lost-banknote claim, i.e. $\textsf{ActiveLostClaim} = \text{``Claim by $pid$ at time $t$''}$, any party $Q$ can challenge that claim by submitting a message $(\textsf{ChallengeClaim}, c, s')$ to the contract, where $s'$ is a proposed new serial number. We say that $c$ is a valid classical certificate for the current value $s$ of $\textsf{serial}$ if $\textsf{verify-certificate}(s,c) = 1$. Such a $c$ can be thought of as a proof that whoever is challenging the claim actually possessed a quantum banknote with serial number $s$, and destroyed it in order to obtain the certificate $c$, and thus that the current active lost-banknote claim is malicious. If $c$ is a valid classical certificate for $s$, then $\textsf{serial}$ is updated to the new value $s'$ submitted by $Q$, who also receives all of the coins deposited in the contract (including the $d_0$ coins deposited by the malicious claim).
    \item If party $P$ has previously submitted a lost-banknote claim, and his claim stays unchallenged for time $t_{tr}$, then party $P$ can send a message $(\textsf{ClaimUnchallenged}, s')$ to the contract, where $s'$ is a proposed new serial number. Then the contract returns to $P$ the $d_0$ coins he initially deposited when making the claim, and updates $\textsf{serial}$ to $s'$.
    \item Any party $P$ can submit to the contract a message $(\textsf{RecoverCoins}, c)$. If $c$ is a valid classical certificate for the current value of $s$ of $\textsf{serial}$, then the contract releases to $P$ all the coins currently deposited in the contract. This allows party $P$ to ``convert'' back his quantum banknote into coins.
\end{itemize}

Next, we will formally define banknote-contracts, and then formally describe all of the protocols that make the payment system.

\begin{figure}[H]
\rule[1ex]{16.5cm}{0.5pt}\\
$\phi_{\$}\left(pid, w, t, (\textsf{serial}, \textsf{ActiveLostClaim}), d\right)$ takes as input strings $pid$ and $w$, where $pid$ is meant to be the PID of some party $P$, and we refer to $w$ as the ``witness'', $t \in \mathbb{N}$ denotes the ``current time'' mantained by $\mathcal{F}_{Ledg}$, $(\textsf{serial}, \textsf{ActiveLostClaim})$ is the current value of the state variable, and $d \in \mathbb{N}$ is the number of coins that are being deposited to the smart contract with the current message. $\phi_{\$}$ has hardcoded parameters: $d_0 \in \mathbb{N}$ the amount of coins needed to file a claim for a lost money state, $t_{tr} \in \mathbb{N}$ the time after which an unchallenged claim can be settled ($d_0$ and $t_{tr}$ are fixed constants agreed upon by all parties, and they are the same for all banknote-contracts), $l(\lambda)$ the length of a certificate in the lightning scheme, a description of \textsf{verify-certificate}. The circuit $\phi_{\$}$ outputs new values for the state variables and an amount of coins as follows:

On input $(pid, w, t, (\textsf{serial} = s, \textsf{ActiveLostClaim}), d)$, $\phi_{\$}$ does the following:
\begin{itemize}
\item If $\textsf{ActiveLostClaim} = \text{``No active claim''}$:
\begin{itemize}
\item If $w = \textsf{BanknoteLost}$ and $d = d_0$, then $\phi$ outputs  $\big((\textsf{serial} = s, \textsf{ActiveLostClaim} = \text{``Claim by $pid$ at time $t$''}), 0\big)$ (to symbolize that at time $t$ party with PID $pid$ has claimed to have lost the money state with serial number $s$, and that zero coins are being released).
\item If $w = (\textsf{RecoverCoins}, c)$, where $c \in \{0,1\}^{l(\lambda)}$ and $\textsf{verify-certificate}(s,c) = 1$, then $\phi$ outputs $\big( (\textsf{serial} = \perp, \textsf{ActiveLostClaim} = \perp), \textnormal{``all coins''}\big)$
\end{itemize}
\item If $\textsf{ActiveLostClaim} = \text{``Claim by $pid'$ at time $t_0$''}$ for some $pid', t_0$: 
\begin{itemize}
\item If $w = (\textsf{ChallengeClaim}, c, s')$,  where $c \in \{0,1\}^{l(\lambda)}$ and $\textsf{verify-certificate}(s,c) = 1$, and $s' \in \{0,1\}^{\lambda}$ then $\phi$ outputs  $\big((\textsf{serial} = s', \textsf{ActiveLostClaim} = \text{``No active claim''}), d_0 \big)$.
\item If $w = (\textsf{ClaimUnchallenged}, s')$, $pid=pid'$ and $t-t_0 > t_{tr}$, then $\phi$ outputs  $\big((\textsf{serial} = s', \textsf{ActiveLostClaim} = \text{``No active claim''}), d_0 \big)$.
\end{itemize}
\end{itemize}

\rule[2ex]{16.5cm}{0.5pt}\vspace{-.5cm}
\caption{Circuit $\phi_{\$}$ for banknote-contracts}
  \label{fig: circuit phi}
  
\end{figure}

\begin{definition}(Banknote-contract)
\label{def: banknote-contract}
A banknote-contract, is a smart contract on $\mathcal{F}_{Ledg}$ specified by parameters of the following form: $I = \{pid\}$ for some $pid \in [n]$, $D = \{(pid, d_{pid})\}$ for some $d_{pid} \in \mathbb{N}$, $\textsf{st}_0 = (s, \text{``No active claim''})$ for some $s \in \{0,1\}^{\lambda}$, and circuit $\phi = \phi_{\$}$, where $\phi_{\$}$ is defined as in Fig. \ref{fig: circuit phi}.
\end{definition}

For convenience, we denote by $\textsf{serial}$ and $\textsf{ActiveLostClaim}$ respectively the first and second entry of the state variable of a banknote-contract.


\paragraph{Generating valid quantum banknotes} We describe the formal procedure for generating a valid quantum banknote.



\begin{figure}[H]
\rule[1ex]{16.5cm}{0.5pt}\\
Protocol carried out by some party $P$ with PID $pid$.\\

Input of $P$: An integer $d$ such that $pid.\textsf{coins} > d$ in $\mathcal{F}_{Ledg}$ ($d$ is the ``value'' of the prospective banknote).
\begin{itemize}
\item Run $(\ket{\psi},s) \leftarrow \textsf{gen-bolt}$. 
\item Send $\big(\textsf{AddSmartContract},  \textsf{Params} \big)$ to $\mathcal{F}_{Ledg}$, where $\textsf{Params} = (\{pid\}, \{(pid, d)\}, \phi, (s, \text{``No active claim''}))$. Upon receipt of a message of the form (\textsf{RecordedContract}, $ssid$), send the message (\textsf{InitializeWithCoins}, $ssid$, \textsf{Params}) to $\mathcal{F}_{Ledg}$.
\end{itemize}

\rule[2ex]{16.5cm}{0.5pt}\vspace{-.5cm}
\caption{Generating a valid banknote}
  \label{fig: protocol gen banknote}
  
\end{figure}

\paragraph{Making a payment} We describe formally the protocol for making a payment in Fig. \ref{fig: protocol making a payment}. Informally, the protocol is between a party $P$, the payer, and a party $Q$, the payee. In order to pay party $Q$ with a bolt whose serial number is $s$, party $P$ sends the valid bolt to party $Q$, the payee, together with the $ssid$ of a smart contract with $\textsf{serial} =s$. Party $Q$ verifies that $ssid$ corresponds to a banknote-contract with $\textsf{serial} = s$, and verifies that the banknote passes verification and has serial number $s$. 

\begin{figure}[H]
\rule[1ex]{16.5cm}{0.5pt}
The protocol is between some party $P$ with PID $pid$(the payer) and a party $Q$ with PID $pid'$ (the payee):\\

Input of $P$: $\ket{\Psi}$, a valid bolt with serial number $s$. $ssid$ the session identifier of a smart contract on $\mathcal{F}_{Ledg}$ such that $ssid.\textsf{state} = (s,\text{``No active claim''})$, and  $ssid.\textsf{coins} = d$. 
\begin{itemize}
\item $P$ sends $(\ket{\Psi},s,ssid,d)$ to $Q$.
\item $Q$ sends a message (\textsf{RetrieveContract}, $ssid$) to $\mathcal{F}_{Ledg}$. Upon receiving a message (\textsf{RetrieveContract}, $ssid$, $z$) from $\mathcal{F}_{Ledg}$ (where $z = (ssid.\textsf{Params}, ssid.\textsf{state}, ssid.\textsf{coins})$ if $P$ is honest), $Q$ does the following: 
\begin{itemize}
\item If $z = (\textsf{Params}, (s, \text{``No active claim''}), d)$, then $Q$ checks that the parameters $\textsf{Params}$ are of the form of a banknote-contract (from Definition \ref{def: banknote-contract}). If so, runs $\textsf{verify-bolt}(\ket{\Psi},s)$ and checks that the outcome is $1$. If so, sends the message $\textsf{accept}$ to $P$.
\item Else, $Q$ sends the message \textsf{reject} and the state $\ket{\Psi}$ back to $P$.
\end{itemize} 
\end{itemize}

\rule[2ex]{16.5cm}{0.5pt}\vspace{-.5cm}
\caption{Protocol for making and verifying a payment}
  \label{fig: protocol making a payment}
  
\end{figure}

\paragraph{Recovering lost banknotes} 
\label{sec: recovering}
As much as we can hope for experimental progress in the development of quantum memories, for the foreseeable future we can expect quantum memories to only be able to store states for a time on the order of days. It is thus important that any payment system involving quantum money is equipped with a procedure for users to recover the value associated to quantum states that get damaged and become unusable. Either users should be able to ``convert'' quantum money states back to coins on the blockchain, or they should be able, upon losing a quantum banknote, to change the serial number state variable of the associated smart contract to a new serial number (presumably of freshly generated quantum banknote). Here, we describe a protocol for the latter. After this, we will describe a protocol for the former.

Informally, a party $P$ who has lost a quantum banknote with serial number $s$ associated to a smart contract with session identifier $ssid$, makes a ``lost banknote claim'' at time $t$ by depositing a number of coins $d_0$ to that banknote-contract. Recall the definition of banknote-contracts from Definition \ref{def: banknote-contract}, and in particular of the circuit $\phi_{\$}$:
\begin{itemize}
\item If party $P$ is honest, then after a time $t_{tr}$ has elapsed, he will be able to update the state variable $\textsf{serial}$ of the banknote-contract from $s$ to $s'$ (where $s'$ is presumably the serial number of a new valid bolt that party $P$ has just generated).
\item If party $P$ is dishonest, and he is claiming to have lost a banknote with serial number $s$ that someone else possesses, then the legitimate owner can run $\textsf{gen-certificate}(\ket{\psi},s)$ where $\ket{\Psi}$ is the legitimate bolt, and obtain a valid certificate $c$. He can then send $c$ to the contract and a new serial number $s'$ (presumably of a freshly generate bolt) and obtain $d_0$ coins from the contract (the $d_0$ coins deposited by $P$ in his malicious claim).
\end{itemize}
We describe the protocol formally in Fig. \ref{fig: protocol making a claim}.

One might wonder whether, in practice, an adversary can instruct a corrupt party to make a ``lost banknote claim'', and then intercept an honest party's classical certificate $c$ before this is posted to the blockchain, and have a second corrupt party post it instead. This attack would allow the adversary to ``steal'' the honest party's value. Or alternatively, an adversary could monitor the network for lost-banknote claims by honest parties, and whenever he sees one he will delay this claim, and instruct a corrupt party to make the same claim so that it is registered first on the ledger. In our analysis, we do not worry about such attacks, as we assume access to the ideal functionality $\mathcal{F}_{Ledg}$, which, by definition, deals with incoming messages in the order that they are received. We also assume in our adversarial model, specified more precisely in Section \ref{sec: security}, that the adversary does not have any control over the delivery of messages (and their timing). If one assumes a more powerful adversary (with some control over the timing of delivery of messages), then the first issue can still be resolved elegantly. The second issue has a satisfactory resolution but it is trickier to analyze formally. We discuss this in more detail in Section \ref{sec: practical issues}. 

\begin{figure}[H]
\rule[1ex]{16.5cm}{0.5pt}
Protocol carried out by party $P$ with PID $pid$ for changing the serial number of a smart contract.\\

$P$'s input: $s$ the serial number of a (lost) quantum banknote. $ssid$ the session identifier of a banknote-contract such that $ssid.\textsf{state} = (s, \text{``No active claim''})$.

\begin{itemize}
\item $P$ sends $(\textsf{Trigger}, ssid, \textsf{BanknoteLost}, d_0$) to $\mathcal{F}_{Ledg}$. This updates $ssid.\textsf{state}$ to $(s, \text{``Active claim by $pid$ at time $t$''})$
(where $t$ is the current time mantained by $\mathcal{F}_{Ledg}$), and deposits $d_0$ coins into the contract.
\item After time $t_{tr}$, $P$ sends $(\textsf{Trigger}, ssid, (\textsf{ClaimUnchallenged}, s'), 0)$ to $\mathcal{F}_{Ledg}$. If $P$ was honest then $ssid.\textsf{state}$ is updated to $(s', \text{``No active claim''})$, and $d_0$ coins are released to $P$.
\end{itemize}

\rule[2ex]{16.5cm}{0.5pt}\vspace{-.5cm}
\caption{Protocol for changing the serial number of a smart contract}
  \label{fig: protocol making a claim}
  
\end{figure}

Next, we give a protocol carried out by all parties to prevent malicious attempts at changing the state variable $\textsf{serial}$ of a smart contract. Informally, this involves checking the blockchain regularly for malicious attempts at filing lost-banknote claims. 

Recall that $t_{tr}$ was defined in Definition \ref{def: banknote-contract}.

\begin{figure}[H]
\rule[1ex]{16.5cm}{0.5pt}

Protocol carried out by a party $P$ to prevent malicious attempts at changing the state variable $\textsf{serial}$ of a smart contract.\\

Input of $P$: A triple $(\ket{\Psi}, s, ssid)$, where $\ket{\Psi}$ is a quantum banknote with serial number $s$, and $ssid$ is the session identifier of a banknote-contract such that $ssid.\textsf{state} = (s, \text{``No active claim''})$ \\

At regular intervals of time $t_r - 1$, do the following: 
\begin{itemize}
\item Send a message (\textsf{RetrieveContract}, $ssid$) to $\mathcal{F}_{Ledg}$. Upon receiving a message (\textsf{RetrieveContract}, $ssid$, $z$) from $\mathcal{F}_{Ledg}$, if $z =(\textsf{Params}, (s, \text{``Claim by $pid'$ at time $t$ ''}), d)$ for some $pid', t, d$ and for some banknote-contract parameters $\textsf{Params}$:
\begin{itemize}
    \item Run $c \leftarrow \textsf{gen-certificate}(\ket{\Psi}, s)$.
    \item Sample $(\ket{\Psi'},s') \leftarrow \textsf{gen-bolt}$.
    \item Send $(\textsf{Trigger},ssid, (\textsf{ChallengeClaim}, c, s'), 0)$ to $\mathcal{F}_{Ledg}$. (If $P$ was honest, this updates $ssid.\textsf{state} \leftarrow (s', \text{``No active claim''})$ and releases $d_0$ coins to $P$).
\end{itemize} 
\end{itemize}

\rule[2ex]{16.5cm}{0.5pt}\vspace{-.5cm}
\caption{Protocol for preventing malicious attempts at changing the state variable $\textsf{serial}$ of a smart contract. }
  \label{fig: protocol preventing}
  
\end{figure}

\paragraph{Trading a quantum banknote for coins:} Finally, we describe a protocol for trading a quantum banknote to recover all the coins deposited in its associated banknote-contract. 

\begin{figure}[H]
\rule[1ex]{16.5cm}{0.5pt}

Protocol carried out by a party $P$.\\

Input of $P$: A tuple $(\ket{\Psi}, s, ssid, d)$, where $\ket{\Psi}$ is a quantum banknote with serial number $s$, and $ssid$ is the session identifier of a banknote-contract such that $ssid.\textsf{state} = (z, \text{``No active claim''})$ and $ssid.\textsf{coins} = d$.\\ 

\begin{itemize}
\item Run  $c \leftarrow \textsf{gen-certificate}(\ket{\Psi},s)$. 
\item Send message $(\textsf{Trigger}, ssid, (\textsf{RecoverCoins}, c), 0)$ to $\mathcal{F}_{Ledg}$. This releases $d$ coins to $P$.
\end{itemize}

\rule[2ex]{16.5cm}{0.5pt}\vspace{-.5cm}
\caption{Protocol for trading a quantum banknote for coins.}
  \label{fig: protocol trading banknotes}
  
\end{figure}

\section{Security}
\label{sec: security}

We first specify an adversarial model. Security with respect to this adversarial model is formally captured by Theorem \ref{thm: security}. At a high-level, Theorem \ref{thm: security} establishes that, within this adversarial model, no adversary can increase his ``value'' beyond what he has legitimately spent or received to and from honest parties. This captures, for example, the fact that the adversary will not be able to double-spend his banknotes, or successfully file a ``lost banknote claim'' for banknotes he does not legitimately possess.

\paragraph{Adversarial model} 

We assume that all the messages of honest parties are sent using the ideal functionality for authenticated communication $\mathcal{F}_{auth}$, and that the adversary sees all messages that are sent (in UC language, we assume that the adversary is activated every time a party sends a message) but has no control over the delivery of messages (whether they are delivered or not) and their timing. Our payment system can be made to work also if we assume that the adversary can delay delivery of honest parties' messages by a fixed amount of time (see the remark preceding Fig. \ref{fig: protocol making a claim} for more details), but, for simplicity, we do not grant the adversary this power.

The adversary can corrupt any number of parties, and it may do so adaptively, meaning that the corrupted parties are not fixed at the start, but rather an honest party can become corrupted, or a corrupted party can return honest, at any point. The process of corruption is modeled analogously as in the original UC framework, where the adversary simply writes a \textsf{corrupt} message on the incoming tape of an honest party, upon which the honest party hands all of its information to the adversary, who can send messages on the corrupted party's behalf. Our setting is slightly more involved in that corrupted parties also possess some quantum information, in particular the quantum banknotes. We assume that when an adversary corrupts a party he takes all of its quantum banknotes. Importantly, we assume that these are not returned to the party once the party is no longer corrupted. It might seem surprising that we do not upper bound the fraction of corrupted parties. Indeed, such a bound would only be needed in order to realize securely the ideal functionality $\mathcal{F}_{Ledg}$ (any consensus-based realization of $F_{Ledg}$ would require such a bound). Here, we assume access to such an ideal functionality, and we do not worry about its secure realization. Naturally, when replacing the ideal functionalities with real-world realizations one would set the appropriate bound on the corruption power of the adversary, but we emphasize that our schemes are independent of the particular real-world realization. Note that we do not fix a set of parties at the start, but rather new parties can be created (see below for more details).

We assume that (ITMs of) honest parties run the code $\pi$. This represents the ``honest'' code which executes the protocols from Section \ref{sec: main} as specified. The input to $\pi$ then specifies when and which protocols from Section \ref{sec: main} are to be executed. As part of $\pi$, we specify that, upon invocation, a party sends a message \textsf{AddParty} to $\mathcal{F}_{Ledg}$ to register itself. We also specify as part of $\pi$ that an honest party runs the protocol of Fig. \ref{fig: protocol preventing} (to prevent malicious claims for lost banknotes). Moreover, for notational convenience, we specify as part of $\pi$ that each party maintains a local variable $\textsf{banknoteValue}$, which keeps track of the total value of the quantum banknotes possessed by the party. $\textsf{banknoteValue}$ is initialized to $0$, and updated as follows. Whenever a party $P$ successfully receives a quantum banknote (i.e. $P$ is the payee in the protocol from Fig. \ref{fig: protocol making a payment} and does not abort) of value $d$ (i.e. the associated smart contract has $d$ coins deposited), then $P$ updates $\textsf{banknoteValue} \leftarrow \textsf{banknoteValue} +d$. Similarly, when $P$ sends a quantum banknote of value $d$, it updates $\textsf{banknoteValue} \leftarrow \textsf{banknoteValue} -d$.
Finally, we specify also as part of $\pi$, that whenever a party that was corrupted is no longer corrupted, it resets $\textsf{banknoteValue} = 0$ (this is because we assumed that quantum banknotes are not returned by the adversary). The following paragraph leads up to a notion of security and a security theorem.

Let $\mathcal{A}$ be a quantum polynomial-time adversary and $\mathcal{E}$ a quantum polynomial-time environment. Consider an execution of $\pi$ with adversary $\mathcal{A}$ and environment $\mathcal{E}$ (see Section \ref{sec: uc} for more details on what an ``execution'' is precisely). We keep track of two quantities during the execution, which we denote as \textsf{AdversaryValueReceived} and \textsf{AdversaryValueCurrentOrSpent} (These quantities are not computed by any of the parties, adversary or environment. Rather, they are just introduced for the purpose of defining security). The former represents the amount of value, coins or banknotes, that the adversary has received either by virtue of having corrupted a party, or by having received a payment from an honest party. The latter counts the total number of coins currently possessed by corrupted parties, as recorded on $\mathcal{F}_{Ledg}$, and the total amount spent by the adversary to honest parties either via coins or via quantum banknotes (it does not count the value of quantum banknotes currently possessed; these only count once they are successfully spent). Both quantities are initialized to $0$, and updated as follows throughout the execution:
\begin{itemize}
    \item[(i)] When $\mathcal{A}$ corrupts a party $P$: let $d$ be the number of coins of $P$ according to the global functionality $\mathcal{F}_{Ledg}$ and $d'$ be $P$'s \textsf{banknoteValue} just before being corrupted. Then, $\textsf{AdversaryValueReceived} \leftarrow \textsf{AdversaryValueReceived} + d + d'$, and $\textsf{AdversaryValueCurrentOrSpent} \leftarrow \textsf{AdversaryValueCurrentOrSpent} + d$.
    \item[(ii)] When a corrupted party $P$ with $d$ coins and $\textsf{banknoteValue} = d'$ ceases to be corrupted and returns honest, $\textsf{AdversaryValueReceived} \leftarrow \textsf{AdversaryValueReceived} - d$.
    \item[(iii)] When an honest party pays $d$ coins to a corrupted party, $\textsf{AdversaryValueReceived} \leftarrow \textsf{AdversaryValueReceived} + d$. Likewise, when an honest party sends a quantum banknote of value $d$ to a corrupted party, through the protocol of Fig. \ref{fig: protocol making a payment}, then (even if the corrupted party does not return \textsf{accept}) 
    $\textsf{AdversaryValueReceived} \leftarrow \textsf{AdversaryValueReceived} + d$.
    \item[(iv)] When $\mathcal{A}$ succesfully spends a quantum banknote of value $d$ to an honest party $P$, i.e. a corrupted party is the payer in the protocol from Fig. \ref{fig: protocol making a payment} and $P$ is the payee and returns \textsf{accept}, or when $\mathcal{A}$ pays $d$ coins to an honest party, then $\textsf{AdversaryValueCurrentOrSpent} \leftarrow \textsf{AdversaryValueCurrentOrSpent} + d$.
    \item[(v)]
    When a corrupted party receives $d$ coins from a banknote-contract, then $\textsf{AdversaryValueCurrentOrSpent} \leftarrow \textsf{AdversaryValueCurrentOrSpent} + d$. Notice that this can happen only in two ways: $\mathcal{A}$ successfully converts a quantum banknote of value $d$ to coins on $\mathcal{F}_{Ledg}$ (via the protocol of Fig. \ref{fig: protocol trading banknotes}), or a corrupted party successfully challenges a \textsf{BanknoteLost} claim (in this case $d = d_0$).
\end{itemize}
Intuitively, if our payment scheme is secure, then at no point in time should the adversary be able to make $\textsf{AdversaryValueCurrentOrSpent} - \textsf{AdversaryValueReceived} > 0$. This would mean that he has successfully spent/stolen value other than the one he received by virtue of corrupting a party or receiving honest payments. The following theorem formally captures this notion of security. First, we denote by $\mathcal{F}_{Ledg}\text{-}\textrm{EXEC}^{(MaxNetValue)}_{\pi, \mathcal{A}, \mathcal{E}}(\lambda, z)$ the maximum value of $\textsf{AdversaryValueCurrentOrSpent} - \textsf{AdversaryValueReceived}$ during an execution of $\pi$ with adversary $\mathcal{A}$ and environment $\mathcal{E}$, with global shared functionality $\mathcal{F}_{Ledg}$.

\begin{theorem}[Security]
\label{thm: security}
For any quantum polynomial-time adversary $\mathcal{A}$ and quantum polynomial-time environment $\mathcal{E}$, 
$$ \Pr[\mathcal{F}_{Ledg}\text{-}\textrm{EXEC}^{(MaxNetValue)}_{\pi, \mathcal{A}, \mathcal{E}}(\lambda, z) > 0] = negl(\lambda).$$
\end{theorem}

The rationale behind considering executions of $\pi$ and quantifying over all possible adversaries and environments is that doing so captures all possible ways in which a (dynamically changing) system of honest parties running our payment system alongside an adversary  can behave (where the adversary respects our adversarial model).

Recall that, in an execution of $\pi$, the environment has the ability to invoke new parties and assign to them new unique PIDs. Since in $\mathcal{F}_{Ledg}$ the PIDs are used to register parties and initialize their number of coins, this means that the environment has the ability to pick the initial number of coins of any new party that it invokes. Moreover, by writing inputs to the parties input tapes, the environment can instruct honest parties to perform the honest protocols from Section \ref{sec: main} in any order it likes. Quantifying over all adversaries and environments, in the statement of Theorem \ref{thm: security} means that the adversary and the environment can intuitively be thought of as one single adversary. The statement of the theorem thus captures security against realistic scenarios in which new parties can be adversarially created with an adversarially chosen number of coins, and they can be instructed to perform the honest protocols of the payment system from Section \ref{sec: main}, in whatever sequence is convenient to the adversary. 

\begin{proof}[Proof of Theorem \ref{thm: security}]
Suppose for a contradiction that there exists $\mathcal{A}$ and $\mathcal{E}$ such that
\begin{equation}
\label{eq: 1}
\Pr[\mathcal{F}_{Ledg}\text{-}\textrm{EXEC}^{(MaxNetValue)}_{\pi, \mathcal{A}, \mathcal{E}}(\lambda, z) > 0] \neq negl(\lambda).
\end{equation}

Then, we go through all of the possible ways that an adversary can increase its net value, i.e. increase the quantity $\textsf{AdversaryValueCurrentOrSpent} - \textsf{AdversaryValueReceived}$: the adversary can do so through actions from items (ii), (iv) and (v) above. Amongst these, it is easy to see that action (ii) never results in $\textsf{AdversaryValueCurrentOrSpent} - \textsf{AdversaryValueReceived} > 0$. Thus, in order for \eqref{eq: 1} to hold, it must be the case that one of the following happens with non-negligible probability within an execution of $\pi$ with adversary $\mathcal{A}$ and environment $\mathcal{E}$.
\begin{itemize}
\item An action from item (iv) resulted in a positive net value for $\mathcal{A}$, i.e. $\textsf{AdversaryValueCurrentOrSpent} - \textsf{AdversaryValueReceived} > 0$. Notice that for this to happen it must be the case that $\mathcal{A}$ has double-spent a banknote, i.e. $\mathcal{A}$ has produced two banknotes with the same serial number that have both been accepted by honest parties in a payment protocol of Fig. \ref{fig: protocol making a payment}, and so they have both passed verification. But then, it is straightforward to see that we can use this adversary, together with $\mathcal{E}$ to construct an adversary $\mathcal{A}'$ that breaks the security of the quantum lightning scheme (i.e. game \textsf{Counterfeit}): $\mathcal{A}'$ simply simulates an execution of protocol $\pi$ with adversary $\mathcal{A}$ and environment $\mathcal{E}$, and with non-negligible probability the adversary $\mathcal{A}$ in this execution produces two banknotes with the same serial number. $\mathcal{A}'$ uses these banknotes to win the security game of quantum lightning. 
\item An action from item (v) resulted in a positive net value for $\mathcal{A}$. Then, notice that for this to happen it must be that either: 

\begin{itemize}
    \item $\mathcal{A}$ has sent a message $(\textsf{Trigger}, ssid, (\textsf{RecoverCoins}, c), 0)$ to $\mathcal{F}_{Ledg}$ for some $ssid$ and $c$ such that $\textsf{verify-certificate}(s,c) = 1$, where $ssid.\textsf{state} = (s, \text{``No active claim''})$, and the last ``make a payment'' protocol (from Fig. \ref{fig: protocol making a payment}) referencing $ssid$ had an honest party as payee which remained honest at least up until after $\mathcal{A}$ sent his message (or the banknote-contract was initialized by an honest user and the banknote was never spent). But then, one of the following must have happened: 
\begin{itemize}
\item $\mathcal{A}$ possessed a bolt $\ket{\Psi}$ with serial number $s$ at some point, before $\ket{\Psi}$ was spent to the honest user. Then, this adversary would have recovered a valid $c$ and also spent a bolt with serial number $s$ successfully to an honest user. But such an $\mathcal{A}$, together with $\mathcal{E}$, can be used to win game $\textsf{Forge-certificate}$ from Definition \ref{def: extra property} with non-negligible probability, with a similar reduction to the one above, thus violating the property of Definition \ref{def: extra property}. 
\item $\mathcal{A}$ recovered $c$ such that $\textsf{verify-certificate}(s,c) = 1$ without ever possessing a valid bolt with serial number $s$. Again, such an adversary could be used, together with $\mathcal{E}$ to win $\textsf{Forge-certificate}$ from Definition \ref{def: extra property}).
\item $\mathcal{A}$ has successfully changed the serial number of contract $ssid$ to $s$ from some previous $s'$ without possessing a bolt $\ket{\Psi}$ with serial number $s'$. This cannot happen since any honest user who possesses the valid bolt with serial number $s'$ performs the protocol of Fig. \ref{fig: protocol preventing}.
\end{itemize}

\item $\mathcal{A}$ has sent a message $(\textsf{Trigger}, ssid, (\textsf{ChallengeClaim}, c), 0)$ to $\mathcal{F}_{Ledg}$ for some $ssid$ such $ssid.\textsf{state} = (s,\text{``Claim by $pid$ at time $t$''})$ for some $s, pid, t$ with $pid$ honest and $c$ such that $\textsf{verify-certificate}(s,c) = 1$. Since $pid$ is honest, he must be the last to have possessed a valid bolt with serial number $s$. Then, there are two possibilities:
\begin{itemize}
\item $\mathcal{A}$ never possessed a valid bolt with serial number $s$, and succeeded in recovering $c$ such that $\textsf{verify-certificate}(s,c) = 1$. Analogously to earlier, this adversary, together with $\mathcal{E}$, can be used to win $\textsf{Forge-certificate}$.
\item $\mathcal{A}$ possessed a bolt $\ket{\Psi}$ with serial number $s$ at some point, before $\ket{\Psi}$ was spent to an honest user. Analogously to earlier, this means such an $\mathcal{A}$ both recovered a $c$ with $\textsf{verify-certificate}(s,c) = 1$ and spent a bolt with serial number $s$ successfully. Such an $\mathcal{A}$ can be used, together with $\mathcal{E}$, to win $\textsf{Forge-certificate}$. 
\end{itemize}

\end{itemize} 

\end{itemize}

\end{proof}
\paragraph{Remark:} The security guarantee of Theorem \ref{thm: security} establishes that an adversary cannot end up with more value than he started with (after taking into account the amount he received from honest parties and the amount he succesfully spent to honest parties). However, we do not analyze formally attacks which do not make the adversary gain value directly, but which ``sabotage'' honest parties, making them lose value. We believe that it should be possible to capture such attacks within our model by modifying the way we keep track of the adversary's value, but we leave this analysis for future work. We discuss and formalize the notion of ``sabotage'' in detail in Section \ref{sec: sabotage}.


\section{Practical issues in a less idealized setting}
\label{sec: practical issues}
The ideal functionality $\mathcal{F}_{Ledg}$ defined in Section \ref{sec: blockchains} does not capture adversaries that are allowed to see messages sent by honest parties to $\mathcal{F}_{Ledg}$ before they are registered on the ledger, and who could try to use this information to their advantage: by definition of $\mathcal{F}_{Ledg}$, messages are processed and registered on the ledger in exactly the order that they are sent, and are not seen by the adversary until they make it onto the ledger. While such a definition of $\mathcal{F}_{Ledg}$ makes for a clear exposition and clean proof of security, it is in practice unrealistic. In typical real-world implementations of blockchains, miners (and hence potentially adversaries) can see a pool of pending messages which have not yet been processed and registered on the blockchain, and can potentially delay the processing of certain messages, while speeding up the processing of others. This makes the system susceptible to the attacks described in the following subsection.

\subsection{Attacks outside of the idealized setting}
\label{sec: possible attacks}
\begin{enumerate}[(i)]
    \item An adversary could file a malicious ``lost banknote claim'' (running the protocol of Fig. \ref{fig: protocol making a claim}) corresponding to a serial number $s$ of a quantum banknote that he does not possess. This would prompt an honest party who possesses a valid quantum banknote with serial number $s$ to publish the corresponding classical certificate $x$, in order to stop the malicious claim. If the adversary could read this message before it is published on the ledger, it could instruct a corrupt party to also publish $x$, and attempt to have this message appear first on the ledger. This would effectively result in the adversary having stolen the honest party's value associated to the serial number $s$.
    \item Suppose an honest party wants to trade their quantum banknote with serial number $s$ (registered on the ledger) for the coins deposited in the corresponding contract. The honesty party executes the protocol of Fig. \ref{fig: protocol trading banknotes}. This includes publishing the classical certificate $x$ associated to $s$. An adversary who sees $x$ before it is registered on the ledger, could instruct a corrupt party to make the same claim for for the coins in the contract associated to $s$. If the corrupt party's message is processed faster than the honest party's message, the adversary has succeeded in stealing the coins deposited in the contract.
    \item Suppose an honest party has lost a valid quantum banknote with serial number $s$ associated to some contract on the ledger. The honest party files a ``lost banknote claim'' by executing the protocol of Fig. \ref{fig: protocol making a claim}. An adversary who hears this before the claim is registered on the ledger could instruct a corrupt party to make a similar claim, and have it appear on the ledger before the honest claim. This would result in the corrupt party obtaining a valid quantum banknote associated to the above contract.
    \item \label{it:sabotage_in_ideal_model} The unforgeability property alone is not enough for \emph{public} quantum money, as it allows sabotage\mpar{\label{mar:motivation_sabotage}}: an attacker might be able to burn other people's money, without a direct gain from the attack. Consider an adversary that wants to harm its competitor. The adversary does not follow the honest protocol that generates the quantum money state. Instead, it could  create a tweaked quantum money which has a noticeable probability to pass verification once, and fail the second time. This way, the adversary could buy some merchandise using this tweaked quantum money. When the merchant will try to pay to others using this quantum money, the verification will fail -- the receiver will run the verification (and this is exactly the second verification), which will cause it to fail. This is not an issue with schemes in which the verification is a projective (or very close to being projective), but, for example, the scheme by Farhi et al.~\cite{FGH+12} is \emph{not} projective.
    
    Indeed, even though our security proof (see Theorem~\ref{thm: security}) guarantees that an adversary cannot end up with more money that he was legitimately given, it does not rule out sabotage. 
\item Similarly to the sabotage attack on the verification of quantum money, an attacker might sabotage the ability to produce a valid certificate from the verified money. 
\end{enumerate}

In the next Section \ref{sec: bolt_to_signature}, we will introduce a novel property of quantum lightning schemes, which we call \textit{bolt-to-signature capability}, we give a provably-secure construction of it. We will employ this in Section \ref{sec: resolution practical} to give a somewhat elegant resolution to issues (i) and (ii) above. In Section \ref{sec: sabotage}, we show that our scheme is secure against sabotage attacks, hence resolving issues (iv) and (v). In \cref{sec: resolution practical} we discuss a practical approach to resolving issue (iii) above, though it does not contain a formal analysis for reasons which are made clear there.

\subsection{Trading a Bolt for a Signature}
\label{sec: bolt_to_signature}
We define a new property of a quantum lightning scheme which we call ``trading a bolt for a signature'' (and we say that a lightning scheme has \textit{bolt-to-signature} capability). This is an extension of the property of ``trading a bolt for a certificate'' defined in Section \ref{sec: lightning}. The known concrete constructions of quantum lightning do not possess this property, but we will show (Fig. \ref{fig: bolt-to-signature}) that a lightning scheme with \textit{bolt-to-certificate} capability can be bootstrapped to obtain a lightning scheme with \textit{bolt-to-signature} capability. We envision that this primitive could find application elsewhere, and is of independent interest. 

We start with an informal description of this property, highlighting the difference between the \textit{certificate} and \textit{signature} properties. Suppose Alice shows Bob a certificate with a serial number $s$. Bob can conclude that Alice cannot hold a bolt with the same serial number. In particular, if she held such bolt, it means she must have measured and destroyed it to produce the certificate. For the \textit{bolt-to-signature} property, we ask the following: 

\begin{itemize}
    \item There is a way for Alice, who holds a bolt with serial number $s$, to produce a signature of any message $\alpha$ of her choice, \textit{with respect to} serial number $s$.
    \item The signature should be verifiable by anyone who knows $s$. 
    \item Just like a certificate, anyone who accepts the signature of $\alpha$ with respect to $s$ can conclude that Alice can no longer hold a quantum money with serial number $s$;
    \item (one-time security) As long as Alice signs a \textit{single} message $\alpha$, no one other than Alice should be able to forge a signature for a message $\alpha' \neq \alpha$ (even though her state is no longer a valid bolt, Alice can still sign more than one message, but the unforgeability guarantee no longer holds).
\end{itemize}

 A quantum lightning scheme with \textit{bolt-to-signature} capability is a quantum lightning scheme with two additional algorithms: $\textsf{gen-sig}$ is a QPT algorithm which receives as input a quantum state, a serial number, and a message of any length, and outputs a classical signature. $\textsf{verify-sig}$ is a PPT algorithm which receives a serial number, a message of any length and a signature, and either accepts or rejects. Thus, we modify the setup procedure of the lightning scheme so that \textsf{QL.Setup} outputs a tuple $(\textsf{gen-bolt}, \textsf{verify-bolt}, \textsf{gen-sig}, \textsf{verify-sig})$.
 
 The definition of the property has a completeness and a soundness part. The latter is formally defined through the following game \textsf{Forge-sig}. The game \textsf{Forge-sig} is similar in spirit to the game for \emph{onetime} security of a standard digital signature scheme -- see, e.g.~\cite[Definition 12.14]{KL14}, ~\cite[Definition 6.4.2]{Gol04}. 
 \begin{itemize}
     \item The challenger runs $(\textsf{gen-bolt}, \textsf{verify-bolt}, \textsf{gen-sig}, \textsf{verify-sig}) \leftarrow \textsf{QL.Setup}(1^\lambda)$ and sends this tuple to $\mathcal{A}$.
     \item $\mathcal{A}$ sends $(\ket{\psi}, s)$ and $\alpha$ to the challenger.
     \item The challenger runs $\textsf{verify-bolt}(\ket{\psi}, s)$. If this rejects, the challenger outputs ``$0$''. Else it proceeds to the next step.
     \item Let $\ket{\psi'}$ be the leftover state. The challenger runs $\sigma \leftarrow \textsf{gen-sig}(\ket{\psi'},s,\alpha)$ and sends $\sigma$ to $\adv$.
     \item $\adv$ returns a pair $(\alpha',\sigma')$.
     \item The challenger checks that $\alpha \neq \alpha'$ and runs $\textsf{verify-sig}(s, \alpha', \sigma')$. If the latter accepts, the challenger outputs ``$1$''.
     \end{itemize}
Let $\textsf{Forge-sig}(\adv,\lambda)$ be the random variable for the outcome of the game.

\begin{definition}[Trading a bolt for a signature]
\label{def: bolt-to-sig}
We say that a quantum lightning scheme has bolt-to-signature capability if the following holds:
\begin{itemize}
    \item[(I)] For every $\alpha$:
    \begin{align}
\Pr[ \textsf{verify-sig}(s, \sigma, \alpha) =  1:  &(\textsf{gen-bolt}, \textsf{verify-bolt}, \textsf{gen-sig}, \textsf{verify-sig}) \leftarrow \textsf{QL.Setup}(1^{\lambda}), \\
&(\ket{\psi}, s) \leftarrow \textsf{gen-bolt}, \\ & \sigma \leftarrow \textsf{gen-sig}(\ket{\psi}, s, \alpha)] = 1-negl(\lambda)
\end{align}
    \item[(II)] For all polynomial-time quantum algorithms $\mathcal{A}$, $$ \Pr[ \textsf{Forge-sig}(\mathcal{A}, \lambda)= 1 ] = negl(\lambda).$$
\end{itemize}
\end{definition}


Informally, the security definition based on game \textsf{Forge-sig} guarantees that:
\begin{itemize}
    \item As long as Alice signs only one message with respect to $s$, no one except her can forge a signature of another message with respect to $s$. This property is very similar to one-time unforgeability for digital signatures, if one views the bolt as a secret key, and the serial number $s$ as a public key. The difference is that the ``secret key'' in this case has meaning beyond enabling the signing of messages: it can be spent as quantum money. This is what make the next property important.
    \item Signing a message destroys the bolt, i.e. it is infeasible to simultaneously produce both a valid signature with respect to a serial number $s$ and a bolt with serial number $s$ which passes verification (an adversary who can succeed at this is easily seen to imply an adversary who wins $\textsf{Forge-sig}$). This property is unique to the quantum lightning setting. It says that signing a message with respect to serial number $s$ inevitably destroys the bolt with serial number $s$. We remark that it is possible to sign more messages with the leftover state, but such state will no longer pass the quantum lightning verification procedure, i.e. it can no longer be spent. One can think of the bolt with serial number $s$ as being ``burnt'' once the owner decides to use it to sign a message.
\end{itemize} 

We are now ready to present our construction of a quantum lightning scheme with bolt-to-signature capability. The construction is based on the hash-and-sign paradigm (see, e.g.,~\cite{KL14}), as well as Lamport signatures~\cite{Lam79} (familiarity with those is helpful, although our presentation is self-contained). For convenience, we use $\mathcal{H}$ to denote a family of fixed-length hash functions, and we use the notation $H \leftarrow \mathcal{H}(1^n)$ to denote an efficent random sampling of a hash function $H$ with output length $n$ from the family $\mathcal{H}$.

\begin{figure}[H]
\rule[1ex]{16.5cm}{0.5pt}
Given: A quantum lightning scheme with bolt-to-certificate capability with setup procedure $\textsf{QLC.Setup}$. A family $\mathcal{H}$ of fixed-length collision-resistant hash functions.

\vspace{2mm}

$\textsf{QLDS.Setup}$: takes as input a security parameter $\lambda$, and outputs a tuple of (descriptions of) algorithms $(\textsf{QLDS.Gen}, \textsf{QLDS.Ver},  \textsf{QLDS.gen-sig}, \textsf{QLDS}.\verifysig)$.
\begin{itemize}
\item Let $n = \poly$. Sample $(\textsf{gen-bolt}, \textsf{verify-bolt}, \textsf{gen-bolt}, \textsf{verify-certificate}) \leftarrow \textsf{QLC.Setup()}$. Sample $H:\{0,1\}^* \rightarrow \{0,1\}^n \leftarrow \mathcal{H}(1^n)$.
\item \textsf{QLDS.Gen}: Run \textsf{gen-bolt} $2n$ times to obtain bolts $(\ket{\psi_i}  \in \mathcal{H}_{\lambda}, s_i)$, for $i=1,..,2n$. Let $\ket{\Psi} = \bigotimes_{i=1,..,2n} \ket{\psi_i} \in \mathcal{H}_{\lambda}^{\otimes 2n}$. Let $s = s_1 || .. || s_{2n}$. Output $(\ket{\Psi}, s)$.
\item \textsf{QLDS.Ver}: Takes as input a state $\ket{\Psi} \in \mathcal{H}_{\lambda}^{\otimes 2n}$. Applies \textsf{QLDS.Ver} to each of the $2n$ factors, and outputs ``accept'' if all $2n$ verifications output ``accept''.
\item \textsf{QLDS.gen-sig}: Takes as input a state $\ket{\Psi} \in \mathcal{H}_{\lambda}^{\otimes 2n}$, a serial number $s$ and a message $\alpha$ of any length. Lets $\beta = H(\alpha) \in \{0,1\}^n$. For $i=1,..,n$:\\
Run $\textsf{gen-certificate}$ on the $(\beta_i\cdot n +i)$-th factor to obtain a certificate $x_i$. Outputs $\sigma = x_1 ||.. ||x_n$.
\item \textsf{QLDS.verify-sig}: Takes as input a serial number $s$, a message $\alpha$, a signature $\sigma$. Parses $s$ as $s = s_1||..||s_{2n}$, $\sigma$ as $\sigma = x_1||..||x_n$. Computes $\beta = H(\alpha)$. For $i = 1,..,n$:\\
Let $r_i = \textsf{verify-certificate}(s_{\beta_i \cdot n + i}, x_i)$. Output ``accept'' if $r_i = 1$ for all $i$.
\end{itemize}

\rule[2ex]{16.5cm}{0.5pt}\vspace{-.5cm}
\caption{Our construction of a quantum lightning scheme with bolt-to-signature capability $QLDS$}
  \label{fig: bolt-to-signature}
\end{figure}

We state and prove the main theorem of this section.

\begin{theorem} If there exists a secure quantum lightning scheme which has a bolt-to-certificate capability, and a family of fixed length quantum-secure collision-resistant hash functions, then there exists a quantum lightning scheme with bolt-to-signature capability.

Specifically, under the assumptions that $\textsf{QLC}$ is a secure quantum lightning scheme with a bolt-to-certificate capability and that $\mathcal{H}$ is a fixed-length family of collision-resistant hash functions, the construction $\textsf{QLDS}$ of Fig. \ref{fig: bolt-to-signature} is a secure quantum lightning scheme with bolt-to-signature capability.
\label{thm:certificate_implies_signature}
\end{theorem}

\begin{proof}
First of all, it is clear that $\textsf{QLDS}$ is still a secure quantum lightning scheme according to Definition \ref{def: lightning basic security}.

The fact that $\textsf{QLDS}$ satisfies the correctness requirement of a quantum lightning scheme with bolt-to-signature capability ($(I)$ in Definition \ref{def: bolt-to-sig}) follows from the correctness property of $\textsf{QLC}$ (namely $(I)$ in Definition \ref{def: extra property}) and that $n$ is at most polynomial in $\lambda$. 


Assume towards a contradiction that there exists a QPT adversary $\adv$ that wins the $\sigforge$ game with non-negligible probability $\epsilon(\secpar)$. We construct an adversary $\mathcal{B}$ that wins the game $\textsf{Forge-certificate}$ (from Definition \ref{def: extra property}) with probability at least $\epsilon(\lambda)$. $\mathcal{B}$ runs as follows:
\begin{itemize}
    \item $\mathcal{B}$ receives a tuple $(\textsf{gen-bolt}, \textsf{verify-bolt}, \textsf{gen-certificate}, \textsf{verify-certificate})$ from the challenger.
    \item $\mathcal{B}$ constructs algorithms $\textsf{gen-sig}$ and $\textsf{verify-sig}$ from $\textsf{gen-certificate}$ and \textsf{verify-certificate}. Sends the tuple $(\textsf{gen-bolt}, \textsf{verify-bolt}, \textsf{gen-sig}, \textsf{verify-sig})$ to $\mathcal{A}$. 
    \item $\mathcal{A}$ returns a pair $(\ket{\psi}, s)$ and a message $\alpha$. Let $\beta = H(\alpha)$. $\mathcal{B}$ simulates the next steps of the challenger in $\textsf{Forge-sig}$ with the following modification: it runs $\textsf{verify-bolt}(\ket{\psi},s)$ but only measures the registeres corresponding to indexes associated with $\beta$ (i.e. $\beta_i \cdot n +i$ for $i \in [n]$). It then runs $\sigma \leftarrow \textsf{gen-sig}(\ket{\psi'}, s, \alpha)$, where $\ket{\psi'}$ is the leftover state after verification. $\mathcal{B}$ sends $\sigma$ to $\adv$, and $\mathcal{A}$ returns a pair $(\alpha', \sigma')$, where $\sigma' = x'_1 || .. || x'_n$.
    \item $\mathcal{B}$ computes $\beta = H(\alpha)$ and $\beta' = H(\alpha')$. If $\beta \neq \beta'$, let $i$ be the first index such that $\beta_i \neq \beta_i'$. $\mathcal{B}$ outputs $(\ket{\psi_{\beta_i' \cdot n +i}}, s_{\beta_i' \cdot n +i})$ and $x'_i$.
\end{itemize}
We analyze the winning probability of $\mathcal{B}$ in game \textsf{Forge-certificate}. With probability at least $\epsilon(\lambda)$ it is $\alpha \neq \alpha'$ (otherwise $\mathcal{A}$ simply loses). Moreover, with overwhelming probability, it must be $\beta \neq \beta'$, otherwise it is immediate such an $\mathcal{A}$ could be used to break collision-resistance of $H$. Hence, with non-negligible probability, there is an index $i$ such that $\beta_i\neq \beta_i'$. Using the definition of \textsf{QLDS.verify-sig}, we deduce that with probability at least $\epsilon$ it must be that $\textsf{verify-certificate}(s_{\beta_i' \cdot n +i}, x'_i) =1$. Moreover, the state $(\ket{\psi_{\beta_i' \cdot n +i}}$ was not measured, and it must pass verification with probability at least $\epsilon(\lambda)$.

\end{proof}

\subsection{Security against Sabotage}
\label{sec: sabotage}
To address Item (\ref{it:sabotage_in_ideal_model}) in \cref{sec: possible attacks}, we define two security games that capture the notion of sabotage; the first is denoted $\sabotagemoney$:
\begin{itemize}
    \item The challenger runs $(\textsf{gen-bolt}, \textsf{verify-bolt}) \leftarrow \textsf{QL.Setup}(\lambda)$ and sends $(\textsf{gen-bolt}, \textsf{verify-bolt})$ to $\mathcal{A}$.
    \item $\adv$ outputs a quantum state $\ket{\psi}$ and sends it to the challenger.
    \item The challenger runs $\textsf{verify-bolt}$ two consecutive times on the quantum state $\ket{\psi}$. 
    \item The adversary wins if the first verification accepts with a serial number $s$ and the second rejects, or accepts with a serial number $s'\neq s$. Let $\sabotagemoney(\adv,\lambda)$ be the random variable that is $1$ if the adversary $\adv$ wins, and is $0$ otherwise.
\end{itemize}
\begin{definition}[Security against sabotage] 
A quantum lighting scheme is secure against sabotage if for every QPT $\adv$ there exists a negligible function $\negl[]$ such that:
\begin{equation}
    \Pr(\sabotagemoney(\adv,\lambda)=1)=\negl
    \label{eq:sabotage_lighning}
\end{equation}
\label{def:sabotage_lightning}
\end{definition}
 The security against sabotage was first defined in the context of quantum money in ~\cite{BS16} (though, the term sabotage was not used).

We extend the notion of sabotage in the natural way for schemes with bolt-to-certificate or bolt-to-signature capability. Our goal is to avoid a scenario in which an adversary gives a user a quantum lightning state which passes verification, but later fails to produce a valid certificate or signature.
We define the following experiment, $\sabotagecertificate$:
\begin{enumerate}
     \item The challenger runs $(\textsf{gen-bolt}, \textsf{verify-bolt},\gencertificate,\verifycertificate) \leftarrow \textsf{QL.Setup}(\lambda)$ and sends that tuple to $\mathcal{A}$.
    \item $\adv$ sends $\ket{\psi},s$ to the challenger.
    \item The challenger runs $\textsf{verify-bolt}(\ket{\psi},s)$. If verification fails, set $r=1$. Otherwise, the challenger uses the post-measured state $\ket{\psi'}$ to generate a certificate $c\gets \gencertificate(\ket{\psi'},s)$, and checks whether it is a valid certificate:   $r \gets \gencertificate(\ket{\psi'},s)$.
    \item $\adv$ wins if $r=0$. Let $\sabotagecertificate(\adv,\lambda)$ be the random variable that is $1$ if the adversary $\adv$ wins, and is $0$ otherwise.
\end{enumerate}
\begin{definition}
A quantum lighting scheme with a bolt-to-certificate capability is secure against sabotage if, in addition to the requirement in \cref{def:sabotage_lightning}, for every QPT $\adv$ there exists a negligible function $\negl[]$ such that:
\begin{equation}
    \Pr(\sabotagecertificate(\adv,\lambda)=1)=\negl
    \label{eq:sabotage_certificate}
\end{equation}
\label{def:sabotage_certificate}
\end{definition}

We suspect that Zhandry's construction based on non-collapsing hash functions, as well as the construction by Farhi et al. (see p.~\pageref{par:farhi_et_al}) does not satisfy the security against sabotage. Fortunately, the construction based on the multi-collision resistance is secure against sabotage:
\begin{prop}
The quantum lighting construction with the bolt to certificate capability discussed in \cref{prop: extra property 2} is secure  against sabotage.
\end{prop}
\begin{proof}
Zhandry's scheme discussed in \cref{prop: extra property 2} has the property that $\ver(s,\cdot)$, is (exponentially close to) a rank-1 projector, and therefore upon one successful verification, it will continue to pass verifications and therefore satisfy~\cref{def:sabotage_lightning}. In fact, it holds even against unbounded adversaries. This rank-1 projector is such that the state that it accepts only has support on $x$'s that are valid certificates. Since the $\gencertificate$ algorithm is simply a measurement in the standard basis, we conclude that \cref{def:sabotage_certificate} holds.
\end{proof}

The $\sabotagesignature$ experiment is defined in an analogous fashion:
\begin{enumerate}
     \item The challenger runs $(\textsf{gen-bolt}, \textsf{verify-bolt},\sign, \verifysig) \leftarrow \textsf{QL.Setup}(\lambda)$ and sends $(\textsf{gen-bolt}, \textsf{verify-bolt},\sign, \verifysig)$ to $\mathcal{A}$.
    \item $\adv$ sends $\ket{\psi}$, and (a document to be signed) $\alpha$ to the challenger.
    \item The challenger runs $\textsf{verify-bolt}(\ket{\psi})$. If verification fails, set $r=1$. If verification accepts the challenger uses the post-measured state $\ket{\psi'}$ to generate a signature of $\alpha$: $\sigma \gets \sign(\ket{\psi'},s,\alpha)$. The challenger runs $r \gets \verifysig(s,\sigma,\alpha)$. 
    \item $\adv$ wins if $r=0$. Let $\sabotagesignature(\adv,\lambda)$ be the random variable that is $1$ if the adversary $\adv$ wins, and is $0$ otherwise.
\end{enumerate}

\begin{definition}
A quantum lighting scheme with a bolt-to-signature capability is secure against sabotage if, in addition to the requirement in \cref{def:sabotage_lightning}, for every QPT $\adv$ there exists a negligible function $\negl[]$ such that:
\begin{equation}
    \Pr(\sabotagesignature(\adv,\lambda)=1)=\negl
    \label{eq:sabotage_signature}
\end{equation}
\label{def:sabotage_signature}
\end{definition}

\begin{prop}
The construction \textsf{QLDS} in \cref{fig: bolt-to-signature} is secure against sabotage, assuming the underlying quantum lightning scheme with a bolt-to-certificate capability \textsf{QLC} which it uses is secure against sabotage. 
\end{prop}
\begin{proof}
Given a QPT adversary $\adv$ that wins the $\sabotagesignature$ experiment with probability $\epsilon$ with respect to the \textsf{QLDS} scheme, we can construct ano adversary $\bdv$ that wins the $\sabotagecertificate$ experiment with probability $\frac{\epsilon(\lambda)}{2n}$. By our assumption that \textsf{QLC} is secure against sabotage, we conclude that $\epsilon(\lambda)$ is necessarily a negligible function. 

The adversary $\bdv$ receives from his challenger $(\gen,\ver,\gencertificate,\verifycertificate)$. $\bdv$ will sample $H:\{0,1\}^*\rightarrow\{0,1\}^n\leftarrow \mathcal{H}^{1^n}$. $\bdv$ will play the role of the challenger in the
$\sabotagesignature$ experiment,  and will also simulate $\adv$. $\adv$ will receive the tuple above and $H$ as part of the setup, and will produce  a state $\ket{\psi}$ over $2n$ registers, serial number $s_1,\ldots,s_n$ a document $\alpha$. 
$\bdv$ will sample $i\in[2n]$ uniformly at random, and send the $i$'th register of $\ket{\psi}$ to his challenger. 

Next we show that indeed $\bdv$ wins with probability at least $\frac{\epsilon}{2n}$.

We know that with probability $\epsilon(\lambda)$, $\adv$'s challenger will verify all the $2n$ states, generate a certificate from $n$ of these states, and at least one of these certificates will verification. Recall that $i$ was sampled uniformly at random. $\bdv$'s challenger preforms exactly the same procedure -- the challenger generates a certificate from the state it receives and checks whether it it passes the verification as a certificate. The probability that $i$ is one of the failed certificates is therefore at least $\frac{\epsilon}{2n}$.
\end{proof}

\subsection{A resolution of the practical issues}
\label{sec: resolution practical}
We first employ a quantum lightning scheme with bolt-to-signature capability to resolve issues (i) and (ii) of Section \ref{sec: possible attacks}.

We make a simple modification to the protocols of Fig.~\ref{fig: protocol preventing} (preventing and challenging malicious ``lost banknote claims'') and Fig.~\ref{fig: protocol trading banknotes} (trading quantum banknotes for coins).

We upgrade the quantum lightning scheme with bolt-to-certificate capability to one with bolt-to-signature capability. To deal with issue (i), we make two modifications to our payment system of Section \ref{sec: main}.
\begin{itemize}
    \item We modify the protocol of Fig. \ref{fig: protocol preventing} as follows: when party $P$ notices a malicious ``lost banknote claim'' for a serial number $s$ associated to a quantum banknote $\ket{\psi}$ that he possesses, he does not simply compute the classical certificate associated to $s$ and send it in clear to $\mathcal{F}_{Ledg}$. Rather, $P$ computes a signature $\sigma \leftarrow \textsf{gen-sig}(\ket{\psi}, s, \alpha)$ where $\alpha$ is a message saying ``Party $P$ challenges the claim''. 
    \item We modify the definition of banknote-contract (Definition \ref{def: banknote-contract}) so that whenever there is an active lost claim, the coins deposited in a contract with state variable $\textsf{serial} = s$ are released to a party $P$ only upon receipt of a signature $\sigma$ with respect to $s$ of a message ``Party $P$ challenges the claim''.
\end{itemize}

We do not give a formal proof of security of the modified scheme, as this would require first defining formally the modified ledger functionality. Instead we argue informally: it is straightforward to see that the attack described in point (i) of Section \ref{sec: possible attacks} is impossible. Any adversary that is able to carry out that attack could be used to create an adversary that is able to forge a signature with respect to a serial number $s$ of a bolt that they do not possess. This violates the security of the bolt-to-signature property of the lightning scheme.

To deal with issue (ii), we make a similar simple modification: 
\begin{itemize}
    \item We modify the protocol of Fig. \ref{fig: protocol trading banknotes} as follows: in order to trade a valid quantum banknote with serial number $s$ for the coins deposited in a smart-contract with state variable $\textsf{serial} = s$, party $P$ does not simply compute the classical certificate associated to $s$ and send it in clear to $\mathcal{F}_{Ledg}$. Instead, $P$ computes a signature $\sigma \leftarrow \textsf{gen-sig}(\ket{\psi}, s, \alpha)$ where $\alpha$ is a message saying ``Party $P$ wishes to recover the coins in the contract''.
    \item We modify the definition of banknote-contract (Definition \ref{def: banknote-contract}) so that whenever there is no active claim, the coins deposited in a contract with state variable $\textsf{serial} = s$ are released to a party $P$ only upon receipt of a signature $\sigma$ with respect to $s$ of a message ``Party $P$ wishes to recover the coins in the contract''.
\end{itemize}

For a similar reasoning as for point (i), the attack in point (ii) of Section \ref{sec: possible attacks} is no longer possible. 

Dealing with issue (iii) of Section \ref{sec: possible attacks} is trickier. In this case, an honest party $P$ has lost a valid quantum banknote with serial number $s$, and there is no way for $P$ to recover any certificate or signature proving possession of the banknote. The only difference between $P$ and everyone else, as far as owning the coins deposited in the associated contract, is that $P$ knows that the banknote has been damages and lost, and no one else does. The ``lost banknote claim'' protocol of Fig. \ref{fig: protocol making a claim} requires $P$ to send a message to $\mathcal{F}_{Ledg}$ declaring the loss, and the only reason why $P$ is be able to recover the coins deposited in the contract in the idealized setting is that he is the first to make this claim, and no has the ability to challenge it. The situation changes dramatically if we allow the adversary to delay the processing of honest parties' messages to $\mathcal{F}_{Ledg}$ in favour of its own. The adversary could simply notice a ``lost banknote claim'' and take advantage of this by making its own claim and ensuring that it is registered first on the ledger. We propose the following modification to handle this issue:
\begin{itemize}
    \item Instead of directly making a ``lost banknote claim'', a party $P$ posts a commitment to a message of the form ``$P$ is filing a lost banknote claim associated to the smart contract with identifier $ssid$''. The commitment contains a deposit of $d_0$ coins. 
    \item $P$ has to reveal that commitment after no more than $t_0$ blocks.
    \item The coins are released to user $P$ only $t_1$ blocks after the reveal phase, and provided no reclaim was made during that time. 
    \item In case there are two or more reveals to two or more commitments to the same $ssid$ -- the one which was committed to the earliest (i.e., the commitment appears in a an earlier block) receives the coins. 
\end{itemize}

Intuitively, this modification resolves issue (iii). This is because the commitment is hiding, and hence the adversary does not learn anything about the claim it contains prior to the reveal. After the reveal phase, it is too late for the adversary to start the commitment phase. On the other hand, if an adversary simply tries to guess what the content of a claim is (i.e. which quantum banknote was lost), and tries to make a claim of his own, the adversary will most likely lose coins, assuming the frequency at which ``lost banknote claims'' are made is low (recall that making a claim requires staking some coins, which are lost if the claim is successfully challenged).

As much as this resolution to issue (iii) seems to work in theory, the adversary could in practice possess some side information regarding the claim that is hidden in the commitment, and it becomes difficult to model this side information in a way that is both rigorous and faithful to what is possible in practice. We illustrate this with some practical examples.

First, it might be hard to know, for an honest user who lost their quantum banknote, whether the latter was indeed lost, or whether it was stolen. Therefore, an honest user who wrongfully believes that his quantum banknote was lost might end up losing an extra $d_0$ coins, to the thief who stole it, when trying to recover it.

Additionally, if an adversary could guess a serial number of a quantum money state that was or will be lost, and post it before the honest user, the adversary will be the one who will eventually receive these coins. Such a setting might be plausible, for example, during (or even before) power outages.

Another attack vector is the following. Indeed, the commitment does not reveal which quantum banknote was lost. However, each commitment reveals that 1 banknote was lost. Suppose there is a claim for 1743 coins at some time $t$. There is a good chance that these $1743$ coins belong to one person who lost all his $1743$ coins. An adversary that can figure out who owns exactly $1743$ coins that were recently lost, and what are the serial numbers of these coins, can effectively steal these coins. Since our construction does not claim any guarantees about the privacy of users, this information might be readily available to the adversary. Various attacks on the users' privacy are known in the "classical" Bitcoin literature -- see \cref{it:privacy} in \cref{sec: comparison}. 

Therefore, we elect to stir away from formal security claims, and we leave this as a proposed resolution which requires further investigation. 

\section{A practical implementation on Bitcoin and space optimization}
\label{sec: practical implementation on Bitcoin}
In this section, we informally discuss a practical implementation of our payment system using Bitcoin. In Section \ref{sec: main}, we have elected to describe our payment system in a model with access to a ledger functionality that supports stateful universal smart contracts. This was for two main reasons: clarity of exposition and flexibility of the framework. Nevertheless, it is possible to implement our payment system on the Bitcoin blockchain. 

Bitcoin uses a special purpose scripting language for transactions. The basic building block of a a script is called an \emph{opcode}. The opcodes that Bitcoin supports are extremely limited. For example, the opcode OP\_ADD, which adds two inputs, is supported, but even a simple operation such as multiplication or raising a number to a power are \emph{not} supported -- as they are not strictly needed, yet they increase the attack surface (for example, they may be used to preform a memory-based attack). More information about the scripting language and opcodes can be found in Ref.~\cite{NBF+16}. 
The only required adjustment to the Bitcoin protocol needed to implement our payment system is to add two opcodes to the Bitcoin scripting language. These opcodes are utilized once when transforming Bitcoin to quantum banknotes, and once when the quantum banknotes are transformed back to Bitcoin. We provide more detail about this implementation by describing a possible improvement to the space efficiency of our payment system. In Section \ref{sec: bitcoin simple approach}, we informally describe an implementation on Bitcoin of the original payment system. In Section \ref{sec: bitcoin space optimization}, we informally describe a modification that drastically improves space efficiency. 

\subsection{Bitcoin to Quantum Money: The Simple Approach} 
\label{sec: bitcoin simple approach}
In order to use a quantum lightning scheme, of course, we need to first run  $\textsf{QL.Setup}(\secparam)$. In some scenarios, where we can assume some trusted setup, this is not an issue. But in an existing distributed system, such as Bitcoin, this could be considered contentious. One of Zhandry's construction only requires a Common Random String (CRS) -- see Fact~\ref{fac: Zhandry requires CRS}.

In the context of Bitcoin, such a CRS could be generated by a multiple-source randomness extractor, applied on the headers of the first $n$ blocks, where $n$ is determined by the amount of randomness required for the quantum lightning scheme, and the min-entropy each header contains. More details regarding using Bitcoin as a source of randomness can be found in Ref.~\cite{BCG15}\footnote{Though, for our purposes, we do not need a recurring form of randomness usually called a randomness beacon, which is the focus of their work. In this context the source of randomness is only used once.}. Bonneau et al. have argued that the min-entropy in each block header is at least 32 bits~\cite[Table 1]{BCG15}.

To transform $x$ bitcoins associated to a signing key $sk$, to quantum money, the user acts as follows. First, the user mints a quantum lightning state, together with some serial number: $(\ket{\$},s)\gets \mint_{pk}$. Second, the user signs the concatenation\footnote{We use $||$ to denote concatenation.} of a special "quantum money" opcode, the serial number, and the value $y$ using the secret key:
\begin{equation}
     \sigma \gets \sign_{sk}(\text{OP\_BITCOIN\_TO\_QUANTUM\_MONEY} || s || y).
 \end{equation}
 Third, the user propagates the signed message $(\text{OP\_BITCOIN\_TO\_QUANTUM\_MONEY} || s || y,\sigma)$ to the Bitcoin miners. Here $y$ is the value which the quantum money holds, and $x-y$ represents the fee which incentivizes miners to include that message in the blockchain, in line with the current fee mechanism in Bitcoin (which we have not covered in this work -- see ~\cite{NBF+16} for more details). 

A miner would include the message $(\text{OP\_BITCOIN\_TO\_QUANTUM\_MONEY} || s || y,\sigma)$ in their block, as long as (i) the signature is legitimate, i.e., $\verify_{vk}(\text{OP\_BITCOIN\_TO\_QUANTUM\_MONEY} || s || y,\sigma)=1$, (ii) $y<x$, where $x$ is the value of the bitcoins associated with the verification key (bitcoin address) $vk$, and (iii) as long as the miner fee is attractive enough.

Under the above conditions, eventually, a block which contains that transaction would be mined, verified by all other nodes, and would become part of the longest chain.

At this point, the original owner of the $x$ bitcoins does not possess any "standard" bitcoins, since they are considered spent. Instead, she holds the quantum state $(\ket{\$},s)$ with the ``value'' $y$. She could then spend the quantum money by sending it to other others, and it would be treated as having the same value as $y$ bitcoins. To verify the quantum money $(\ket{\psi},s,y)$, the receiver would first check that $\text{OP\_BITCOIN\_TO\_QUANTUM\_MONEY} || s || y$ indeed appears in the blockchain. This step requires storing the blockchain. Any old version that contains the signed message will do. Techniques such as Simplified Payment Verification (SPV) could be used to eliminate the need to store that information, with some downsides in terms of privacy and security~\cite{Nak08,NBF+16}. Then, the receiver would accept the quantum money if $\verify_{pk}(\ket{\psi},s)=1$.  The receiver could in turn spend the received quantum money to other users via the same procedure. 

\subsection{Bitcoin to Quantum Money: Space Optimization} 
\label{sec: bitcoin space optimization}

The approach mentioned above has a limitation. Unlike bitcoin, which is divisible almost infinitely, quantum money is not (see \cref{it:divisibility} in p.~\pageref{it:divisibility} for more details). The most naive workaround is to divide the value of the $x$ bitcoins between several quantum money states, with values $y_1,\ldots,y_n$ so that $\sum_{i\in [n]} y_i < x$. The disadvantage is in terms of space: each serial number is recorded on the blockchain, which causes high fees and somewhat of a scalability problem. 

We now present a much more efficient approach, in which the space required on the blockchain is independent of the number of quantum money states the user creates. Suppose the user wants to split the $x$ bitcoins into $2^n$ quantum money states, each having a value of $2^{-n} y$. The user first creates $2^n$ quantum money states with serial numbers $s_1,\ldots, s_n$, and calculates the Merkle Hash Tree~\cite{Mer80}\footnote{We note that the original motivation of Merkle was very different than the one we use here. See ~\cite[Section 5.6.2]{KL14} and \cite[Section 1.2]{NBF+16} for the specific application we need.} of these serial numbers. The user than signs and publishes only the root of the Merkle tree $r$, the total value $y$ (which must satisfy $y< x$ as before), and $n$. 
Each quantum money state now consists of three parts: the quantum part $\ket{\$}$, the serial number $s_i$, and a Merkle path from $s_i$ to the root $r$. Note that the information recorded on the block-chain is independent of $n$, and one could create very small denominations using this approach. 

Another advantage of this approach is that it allows a slightly weak variant of a Proof of Payment -- see \cref{it:proof_of_payment} in p.~\pageref{it:proof_of_payment}.

One could even take this one step further. Miners could include a Merkle root of the Merkle tree of all existing serial numbers, at the end of every calendar year. This would require users that wish to verify quantum money to come online once every year to store the Merkle root. Every owner of quantum money would calculate the Merkle path from the serial number of the quantum money state, to that root, and append it to their  quantum money state. That would allow these users to verify the quantum money state generated in previous years using only one hash, typically, e.g., 256 bits. Of course, this technique does not reduce the space requirement for quantum money which was generated in a year which has not ended yet. 


\section{Comparing our scheme to classical alternatives}
\label{sec: comparison}
In this section, we discuss several trade-offs between Bitcoin and our payment system based on quantum money. The current roadmap for solving the scalability problem of Bitcoin goes through a second layer solution called the Lightning Network (LN) \footnote{Note that there is no connection between the terms quantum lightning and the lightning network}. This section discusses the trade-offs between the following three different alternatives of employing Bitcoin: (a) a standard Bitcoin transaction, (b) a LN transaction and (c) a quantum money transaction as in our payment system (we will refer to this as QM for short, from now on). Before doing so, we give a brief overview of the LN.  

The LN improves upon Bitcoin, and provides a way to transact off-chain using \emph{bidirectional payment channels} between two parties \cite{PD16} (see also~\cite{BDW18}).
The opening and closing of the channel is done on-chain (i.e., both of these operations require $1$ transaction that would appear on the Bitcoin blockchain) while updating the balance of the channel is done completely off-chain, and thus allows for an unlimited number of transactions which do not affect the limited throughput of the Bitcoin blockchain, and do not incur the long confirmation times of on-chain transactions. At any point, each of the two parties involved can ``close'' the channel, by submitting the most up to date state of the channel to the blockchain. Crucially, the lightning network supports routing. What this means is that the graph induced by all bi-directional payment channels allows Alice to pay Bob, without any on-chain transaction, as long there is a path in the graph from Alice to Bob. The objectives of the LN are to increase the effective transaction throughput of the Bitcoin blockchain (by having most of the transactions happen off-chain), make transactions be effective immediately, and reduce the transaction costs.   
Since early 2018, the LN has been active on the Bitcoin main-net. The capacity of the LN is shown in \cref{fig:capacity_LN}, and as of August 2019, it stores only $0.004\%$ of the total bitcoins in circulation, and is still considered experimental\footnote{For example, as of August 2019, the maximal capacity of a LN channel is $0.16$ bitcoin, see \url{https://github.com/lightningnetwork/lightning-rfc/blob/master/02-peer-protocol.md\#the-open_channel-message}, and the requirements regarding the parameter \textsf{funding\_satoshis}.} . 
\begin{figure}[htb]
    \centering
    \includegraphics[width=\textwidth]{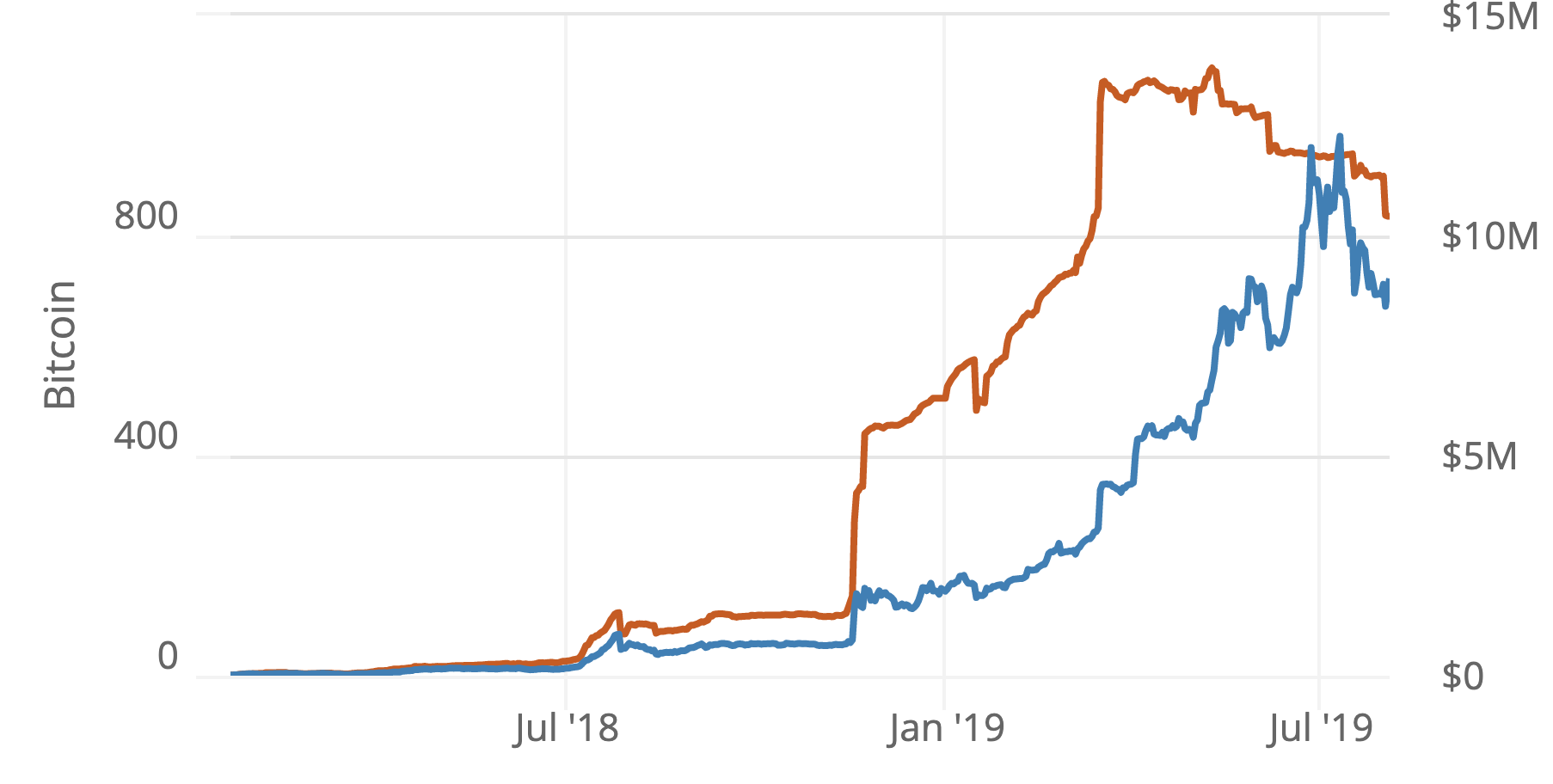}
    \caption{The total capacity of the lightning network, starting on January 2018, brown in BTC, blue in USD. Source: \url{https://bitcoinvisuals.com/ln-capacity}  }
    \label{fig:capacity_LN}
\end{figure}

The following list provides an exhaustive comparison between the three alternative modes of operation mentioned above. It is worth noting that quantum money, used as in our payment system, in many ways resembles (digital) cash banknotes, both in terms of the advantages and the disadvantages, and is thus the closest of the three modes to ideal digital cash. Items 1-\ref{it:communication_availability} present the aspects in which quantum money outperform Bitcoin and the LN, while the disadvantages of quantum money are presented in Items~\ref{it:computational_resources}-\ref{it:proof_of_payment}.
\begin{enumerate}
    \item \label{it:throughput} \textbf{Throughput.} The throughput of Bitcoin is strongly capped. On average, in the first 6 months of 2019, the Bitcoin network had 3.9 transaction per second\footnote{Source: \url{https://www.blockchain.com/charts/n-transactions}.}. This throughput cannot increase dramatically without changing the Bitcoin protocol. 

    To receive bitcoins through the LN, Alice must have an open channel, which she will eventually need to close. This requires at least two on-chain transaction, though, the number of uses is not bounded. In this regard, transforming quantum money to Bitcoin is very similar to opening a quantum channel, and transforming the quantum money back to Bitcoin is similar to closing a channel.

    The balance of a LN channel can be updated, but is always bounded by the transaction that opened that channel. For example, if Alice locks $10$ bitcoins in a channel with Bob, then initially she has $10$ bitcoins and Bob has zero; Alice and Bob could then update it, but Bob could never receive more than $10$ bitcoins using the channel. 

     A user could receive and send quantum money without ever making an on-chain transaction. Quantum money has no limit on the value transferred.

     \item \textbf{Liquidity.} Suppose Alice wants to pay Bob. In Bitcoin, she could do that, assuming she has enough funds, and connection to the Bitcoin network. In the LN, this is not always the case: there needs to be a way to route money among the open channels. Sometimes no such route exists, or is inaccessible (using these channels for routing requires the cooperation of the parties along the channel). This may have an impact on the effective throughput of the LN. A quantum money state can always be sent from the receiver to the sender.  

     \item \textbf{Latency.} The recommended best practice for a Bitcoin transaction is waiting for 6 confirmations, which takes 1 hour on average. 

     Both the LN and QM need one transaction -- in order to open a LN channel, or to transform bitcoin to QM. That part can be done in advance, and is only needed once, but it suffers from the same latency as a Bitcoin transaction. A LN has no inherent latency other than the delays caused by the routing. For example, a single transaction might involve several hops between distant nodes, which takes on the order of a second. QM has a slightly better performance, especially when the sender and receiver are physically close to each other -- the latency in this case is only due to the quantum verification algorithm. Overall, the LN and QM have comparable latencies, which are much better than Bitcoin's.

     \item \textbf{Fees.} Each Bitcoin transaction needs to pay a \emph{fee}, in order to get approved. This fee is collected by the miner, which included the transaction. The average fee per transaction in the first 6 months of 2019 was $1.41$ USD\footnote{Sources: \url{https://www.blockchain.com/charts/n-transactions} and \url{https://www.blockchain.com/charts/transaction-fees-usd}.}. To encourage LN nodes to provide liquidity, the protocol uses routing fees. These fees are expected to be smaller than on-chain transaction fees, but still non-zero. No fees are needed to transact with QM. 

     \item \textbf{Dependence on mining.} The Bitcoin protocol and implicitly, the LN, are based on mining (also known as Proof-of-Work~\cite{DN92}). This approach suffers from two main drawbacks: (a) As of August 2018, Bitcoin mining (also known as Proof-of-Work~\cite{DN92}) consumed slightly under 1\% of the world's electricity~\cite{Nar18}. Mining is required to secure the network. Proof-of-Stake is a competing approach with somewhat different trade-offs~\cite{KRDO17,DGKR18}, with the main advantage that it does not spend so much energy.
     (b) The security model related to mining is convoluted, and is based on incentives, without clear assumptions nor a full security analysis. In particular, it is known that the Bitcoin protocol is \emph{not} incentive compatible~\cite{ES14}.  In PoW and PoS, a corrupted majority can double-spend. In addition, the Bitcoin protocol does not provide finality -- in principle, at any point, a roll-back (known as re-organization) could occur.

     Quantum money does not require PoW or PoS, and provides finality. For example, if all the bitcoins were transformed to quantum money, PoW or PoS would not be needed at all. Of course, coordinating such a change is highly non-trivial. 

    \item \textbf{Privacy.}\label{it:privacy} The users' privacy is not guaranteed in Bitcoin, and certainly not by default. All the transactions that have ever occurred are publicly accessible. The common approach to provide some level of privacy is to use a new bitcoin address per transaction. Yet, there are various techniques to relate these addresses to real-world identities~\cite{RS13,RS14,MPJ16}. There are several commercial services that deanonymize Bitcoin transactions\footnote{E.g., \url{https://www.chainalysis.com}. They claim to have "the biggest database in the world of connections between real world entities and transactions on the blockchain", see~\url{https://youtu.be/yNpNz-FvSYQ?t=154}. }. More recent works tackle this privacy issue -- see ~\cite{Sab13,MGGR13,BCG+14,Poe16} and references therein.

    There is a trade-off between privacy and concurrency in the LN, see~\cite{MMK+17} and references therein. In addition, this aforementioned work explicitly leaves the question of privacy preserving routing for the LN open. 

    Quantum money enjoys superior privacy, as only the sender and receiver take part in the transaction, and the transaction leaves no trace. The privacy of QM is analogous to that of physical banknotes. Bear in mind that banknotes have serial numbers, and these could potentially be traced, although arguably the effect on privacy is negligible. In this sense, coins provide better privacy for the users, since they are indistinguishable. \emph{Quantum Coins} have been formally studied: these are indistinguishable quantum states which provide an analogous level of privacy as physical coins - see~\cite{MS10,JLS18} (improving upon~\cite{TOI03}). Unfortunately, these constructions are for private quantum money (rather than public), and do not constitute a quantum lightning scheme, which is crucial for our construction, and thus cannot be used in our payment system. We leave it as an open question whether the same level of anonymity achieved using quantum coins, or the construction mentioned above, could be achieved in the context of quantum money for Bitcoin.

    \item  \textbf{Risk of Locked funds.} In the LN, the parties lock funds in a channel. When both parties are cooperative, these funds can be easily unlocked, and used in the next block. In case one party is uncooperative, and refuses to close the channel, these funds are effectively locked for some period of time, typically, for a few days. 

    Quantum money does not require users to lock any funds.

     \item \textbf{Connectivity requirements.} A Bitcoin transaction requires communication with the Bitcoin network. At the very least, the receiver needs a communication channel with at least one other Bitcoin node. A LN transaction needs even more resources, since the LN uses source routing -- therefore, the sender has to be well connected to the entire network in order to find the route. A QM transaction only requires quantum communication between the sender and the receiver.


     \item \label{it:communication_availability} \textbf{Liveliness.} For technical reasons which are outside the scope of this work, both parties participating in each channel have to be on-line occasionally, and monitor the Bitcoin network, in order to revoke a transaction, in case the other party cheats. This task can be outsourced to a \emph{Watchtower}, which has to be trusted to perform its job (and of course the watch-tower has to be on-line occasionally). Online availability is not required for QM, if one opts for a version of our payment system in which there is no procedure to recover the value corresponding to lost or damaged quantum banknotes. If one opts to include such a recovery mechanism, then a similar level of online availability as in the LN is required.


     \item \label{it:computational_resources} \textbf{Technological requirements.} Transactions with quantum money require 3 main resources: long term quantum memory to store the quantum money, a universal quantum computer to verify the quantum money, and a quantum network to transmit the quantum money between users.

     \item \textbf{Smart contracts.} Bitcoin provides a scripting language which can be used to design \emph{smart contracts}. Notable examples include \emph{multi-sig transactions}, in which $m$-out-of-$n$ signatures are needed to spend the money, and \emph{atomic cross-chain swap} which allows Alice and Bob to trade two crypto-currencies without trusting each other~\cite{NBF+16}. Both capabilities naturally extend to the LN\footnote{The video in ~\url{https://youtu.be/cBVcgzEuJ7Q} demonstrates a LN atomic swap between Bitcoin and Litecoin.}. Other crypto-currencies such as Ethereum have a richer programming language, which allows constructing Decentralized Applications (DApps) ~\cite{Woo14,But14}. 

     QM does not solve the consensus problem or any variant of it, and does not provide the functionality of smart contracts.

     \item \textbf{Backup.} It is pretty easy to backup a Bitcoin wallet. Typically, all that a user needs is a fairly short string called a \emph{seed}. This seed is used to generate a Hierarchical Deterministic Wallet~\cite{Wui13}.  A backup can be done once, and never needs to be updated in the lifetime of a Bitcoin wallet. LN channels are slightly harder to backup, since the protocol is stateful, and therefore, currently, backing up requires having the most up-to-date state of the channel. 

     By definition, it is impossible to backup a quantum money state. In this work, we proposed a mechanism to \emph{recover} lost banknotes. In essence, a user can claim that her quantum money was lost, by depositing $d$ bitcoins. The user would receive these bitcoins after some period of time (e.g., one year). If the party that claimed that the coins were lost is dishonest, then the person holding the legitimate quantum money state can produce a certificate of this fact, and claim the $d$ bitcoins, in addition to the value of the quantum money that she originally held. To avoid theft, users that want this option available have to be on-line occasionally (in this example, at least once a year).
     
     \item \textbf{Divisibility.}\label{it:divisibility} One of the advantages of Bitcoin and the LN is that any amount, down to $10^{-8}$ bitcoin can be sent, and in principle, even smaller amounts could be used. Quantum money, on the other hand, is \emph{not} divisible. The user must decide, at the time of the quantum minting, what is the denomination of the quantum money, and it remains the same for the lifetime of the quantum money. 

     \item \textbf{Hidden inflation.} \label{it:hidden_inflation} Consider a computationally powerful adversary that attacks the Bitcoin network. Such an adversary could, for example, break the digital signature scheme and steal other people's Bitcoins. Yet, such an adversary couldn't "print" bitcoin from thin air, without others noticing it\footnote{Unless there is some bug that is unrelated to cryptography}. 

    When we use quantum money, the situation is different. A powerful adversary could create new quantum money from thin air, without others being able to notice it. This might be a threat either because of invalid computational assumptions, or flaws in the implementation. This threat is not unique to quantum money. In fact, such a flaw in the implementation occurred in ZCash, a crypto-currency which is based on the ZeroCash protocol~\cite{BCG+14}, see  \url{https://electriccoin.co/blog/zcash-counterfeiting-vulnerability-successfully-remediated}. Inevitably, there is no definitive way to know whether that bug was exploited. 

    \item \textbf{Finality of the security parameters.}\onote{Finish.} One interesting feature of Bitcoin is that the level of security that is achieved can, in principle, be increased. If the level of security seems insufficient due to technological advancement, the protocol may allow users to transition to more secure schemes. This is exactly the case for the proposed post-quantum secure digital signature schemes~\cite{But13,LRS18}. It is the the responsibility (and incentive) of each individual user to transition -- otherwise, her funds might be lost. 

    The quantum money can be used too increase security in the same manner: 
    users with QM in circulation can create new QM with the improved parameters, sign the new serial number using their the bolt-to-signature capability (and by that destroying it), and adding that signed message to the blockchain.
    Yet, the incentives here a slightly different: an adversary could steal the bitcoins of some user, if the security parameters are poorly chosen. In the quantum money setting, the adversary could print money from thin air -- see \cref{it:hidden_inflation}. This makes the system, as a whole, insecure. Therefore, it is advisable to make such a transition mandatory.

    \item \textbf{Optional Transparency.} Bitcoin transactions are publicly available, and organizations or individuals that wish to, can make their accounting book completely transparent. See Bitcoin Improvement Proposal (BIP) 32~\cite{Wui13} and Ref.~\cite[Chapter 4.2]{NBF+16} for more details. 

    Quantum money transactions leave no trace, and therefore it seems impossible to achieve this sort of transparency.


     \item \label{it:proof_of_payment} \textbf{Proof of payment.} Consider the following scenario. You go to a store, and pay for an item. You hand over a valid banknote to the seller. The seller takes it to the bill checking machine, and secretly replaces your valid bill with a fake one. The seller then gives you back the fake money, and blames you for trying to fraud him.

    This kind of attack cannot happen in Bitcoin, if used appropriately. The seller can ask for a signed payment request, and the buyer can then verify the authenticity of that message, using the seller's public key. After payment, the seller cannot argue that the payment was not received -- the buyer can prove that the bitcoins were sent to the seller's address, by showing the payment on the blockchain.

    A Bitcoin payment can be done in such a way that an honest user would have a proof of payment~\cite{AH13}. A similar functionality might be possible to achieve in the LN, though currently, as far as the authors are aware, the LN does not provide such functionality. 

    On the other hand, QM transactions leaves no trace, and proof of payment seems harder to achieve. A possible workaround (which works for the LN as well) could be the following. Suppose Alice wants to send 10 bitcoins worth of quantum money to Bob. Instead of sending it all at a time, she could divide the payment into 100 iterations. In each iteration she would send $0.1$ bitcoin, and expect a digital signature in return approving the payment in return. If Bob fails to provide such a signature, she would abort. The worst case scenario in this case is that she would not have a proof of payment for $0.1$ bitcoins.

 \end{enumerate}









\section{Conclusion}
In this work, we gave the first example of the use of classical smart contracts in conjunction with quantum cryptographic tools. We showed that smart contracts can be combined with quantum tools, in particular quantum lightning, to design a decentralized payment system which solves the problem of scalability of (payment) transactions. There is currently only one known secure construction of quantum lightning, which relies on a computational assumption about multi-collision resistance of certain degree-2 hash functions \cite{zha19}. Finding alternative constructions of quantum lightning, secure under more well-studied computational assumptions, is a very interesting open problem.

Smart contracts have found several applications in classical cryptographic tasks, but their application to quantum cryptographic tasks is virtually unexplored. We hope that this work will ignite future investigations. Some candidate tasks which might potentially benefit from smart contracts are: generation of public trusted randomness, distributed delegation of quantum computation, secure multi-party quantum computation. 

\paragraph{Acknowledgments}
A.C is supported by the Simons Institute for the Theory of Computing. O.S. is supported by the Israel Science Foundation (ISF) grant No. 682/18 and 2137/19, and by the the Cyber Security Research Center at Ben-Gurion University. 

\ifnum\sigconf=1
    \bibliographystyle{ACM-Reference-Format}
\else
    \bibliographystyle{alphaabbrurldoieprint}
\fi

\bibliography{quantum_money_solution_blockchain_scalability}

\appendix
\section{Appendix}
\subsection{Proof of Proposition \ref{prop: extra property 2}}
\label{sec: appendix}
\begin{proof}[Proof of Proposition \ref{prop: extra property 2}]
We assume some familiarity with Zhandry's construction (see section 6 of \cite{zha19} for more details). In his construction, a full bolt is a tensor product of $n$ mini-bolts. A valid mini-bolt with serial number $y$ takes the form $\ket{\Psi}^{\otimes (k+1)}$ where $\ket{\Psi}$ is a superposition of pre-images of $y$ under a certain function $H$ (this is called $f_{\mathcal{A}}$ in Zhandry's paper, and we do not go into the details of what this function is). The serial number of the full bolt is the concatenation of the serial numbers of the mini-bolts. \textsf{verify-bolt} has $H$ hardcoded in its description. The computational assumption under which Zhandry's construction is proved secure is that $H$ is $(2k+2)$-multi-collision resistant, i.e. it is hard to find $2k+2$ colliding inputs (for this particular function it is easy to find $k+1$ on the other hand). 

Similarly to the proof of Proposition \ref{prop: extra property 1}, we define $\textsf{gen-certificate}$ to be the QPT algorithm that measures each mini-bolt in the computational basis and outputs the concatenation of the outcomes. We define \textsf{verify-certificate} to be the deterministic algorithm which receives as input a serial number $(y_1,..,y_n)$, and the concatenation of $(z^{(i)}_{1},.., z^{(i)}_{k+1})$ for $i=1,..,n$, and checks that $H(z^{(i)}_{1}) = .. = H(z^{(i)}_{k+1}) = y_i$ for all $i$.

It is clear that $(I)$ holds.

For property $(II)$, similarly to the proof of Proposition \ref{prop: extra property 1}, we can construct an adversary $\mathcal{A}'$ that breaks the $(2k+2)$-multi-collision resistance of $H$ from an adversary $\mathcal{A}$ that wins game $\textsf{Forge-certificate}$ with non-negligible probability: $\mathcal{A}'$ runs $(\textsf{gen-bolt}, \textsf{verify-bolt}) \leftarrow \textsf{QL.Setup}(1^\lambda)$ and sends $(\textsf{gen-bolt}, \textsf{verify-bolt},\textsf{gen-certificate}, \textsf{verify-certificate})$ to $\mathcal{A}$, where the latter two are defined as above in terms of the function $H$ hardcoded in \textsf{verify-bolt}. $\mathcal{A}$ returns $c$ which is parsed as $(x_1,..,x_n)$, where each $x_i$ is a $(k+ 1)$-tuple, and $\ket{\psi}$. $\mathcal{A}'$ then measures each of the $n$ registers of $\ket{\psi'}$ to get $n$ $(k+1)$-tuples $ (x'_1,..,x'_n)$. If there is some $i$ such that $(x_i, x_i')$ is a $(2k+2)$-collision, then $\mathcal{A}'$ outputs this. We claim that with non-neglibile probability $\mathcal{A}'$ outputs a $(2k+2)$-collision: in fact, since $\mathcal{A}$ wins $\textsf{Forge-certificate}$ with non-negligible probability, then $\ket{\psi'}$ must pass $\textsf{verify-bolt}$ with non-negligible probability; from the analysis of Zhandry's proof, we know that any full bolt that passes verification with non-negligible probability must be such that most mini-bolts have non-negligible weight on most pre-images (in fact they should be close to uniform superpositions over all pre-images). Indeed, in Section 6.3, p. 435, Zhandry argues:
\begin{quote}
    Conditioned on acceptance, by the above arguments the resulting mini bolts
must all be far from singletons when we trace out the other bolts. This means
that if we measure the mini-bolts, the resulting superpositions will have high
min-entropy.
\end{quote}
\mpar{\label{mar:prop2_quote}}
\end{proof}

\subsection{Standard Security Definitions} 
\label{sec:standard_security_definitions}
\begin{definition}[Digital Signature Scheme]\label{def:digital_signature}
A digital signature scheme consists of 3 PPT algorithms $\keygen,\ \sign$ and $\verify$. 
The scheme is complete if the following holds.
When a document is signed using the private key, the signature is accepted by the verification algorithm using the public key. Formally, for every $\alpha\in \{0,1\}^{*}$:
\begin{equation}
\Pr\left[(vk,sk) \gets \keygen(\secparam); \sigma \gets\sign_{sk}(\alpha); r \gets \verify_{vk}\left(\alpha,\sigma) \right): r = 1 \right] = 1
\label{eq:ds_correctness}
\end{equation}

The scheme satisfies Post-Quantum Existential Unforgeability under a Chosen Message
Attack (PQ-EU-CMA) if the following holds. Consider a QPT adversary with the capability of adaptively requesting
documents to be signed by a signing oracle. The scheme is secure if such an adversary
cannot generate a signature for any fresh document -- a document which the adversary did not ask the oracle to sign. Formally, for every QPT adversary $\adv$ there exists a negligible function $\negl[]$ such that
\begin{equation}
\Pr\left[
  (vk,sk) \gets \keygen(\secparam); (\alpha,\sigma) \gets \adv^{\sign_{sk}}(\secparam,vk): \verify_{vk}(\alpha,\sigma) = 1 \wedge \alpha \notin Q_{\adv}^{\sign_{sk}}\right] \leq \negl,
\label{eq:unforgeability_digital_signature} 
\end{equation} 
where $\adv^{\sign_{sk}}$ is a QPT algorithm with access to the signing
oracle, and $Q_{\adv}^{\sign_{sk}}$ is the set of queries it made to the oracle.
\end{definition}

\fi

\end{document}